 \documentclass[5p,times]{elsarticle}

\usepackage{amssymb}
\usepackage{lipsum}

\usepackage{cite}
\usepackage{dblfloatfix}
\usepackage{amsmath,amssymb,amsfonts}
\usepackage{tabularray}
\usepackage{romannum}
\usepackage{ragged2e}
\usepackage{amsthm}
\usepackage{algorithm} 
\usepackage{algpseudocode}

\usepackage{mdframed,lipsum}
\usepackage{graphicx}
\usepackage{textcomp}
\usepackage{xcolor}
\usepackage{theoremref}
\makeatletter
\newcommand*{\rom}[1]{\expandafter\@slowromancap\romannumeral #1@}
\DeclareMathOperator{\EX}{\mathbb{E}}
\makeatother
\newtheorem{theorem}{Theorem}
\newtheorem{proposition}{Proposition}
\newtheorem{problem}{Problem}

\newtheorem{definition}{Definition}
\newtheorem{example}{Example}
\newtheorem{strategy}{Strategy}

\usepackage{enumitem}

\newlength\mylen
\def\BibTeX{{\rm B\kern-.05em{\sc i\kern-.025em b}\kern-.08em
    T\kern-.1667em\lower.7ex\hbox{E}\kern-.125emX}}

\journal{Future Generation Computer Systems}

\begin{document}

\begin{frontmatter}



\title{On Scaling LT-Coded Blockchains in Heterogeneous Networks and their Vulnerabilities to DoS Threats\\
}


\author[first]{Harikrishnan K}
\author[first]{J. Harshan}
\affiliation[first]{organization={Department of Electrical Engineering, Indian Institute of Technology Delhi},
            city={Delhi},
            postcode={110016}, 
            country={India}}
\author[second]{Anwitaman Datta}
\affiliation[second]{organization={College of Computing and Data Science, Nanyang Technological University},
            postcode={639798}, 
            country={Singapore}}

\begin{abstract}
Coded blockchains have acquired prominence as a promising solution to reduce storage costs and facilitate scalability. Within this class, Luby Transform (LT) coded blockchains are an appealing choice for scalability owing to the availability of a wide range of low-complexity decoders. In the first part of this work, we identify that traditional LT decoders like Belief Propagation and On-the-Fly Gaussian Elimination may not be optimal for heterogeneous networks with nodes that have varying computational and download capabilities.
To address this, we introduce a family of hybrid decoders for LT codes and propose optimal operating regimes for them to recover the blockchain at the lowest decoding cost.
While LT coded blockchain architecture has been studied from the aspects of storage savings and scalability, not much is known in terms of its security vulnerabilities. Pointing at this research gap, in the second part, we present novel denial-of-service threats on LT coded blockchains that target nodes with specific decoding capabilities, preventing them from joining the network. Our proposed threats are non-oblivious in nature, wherein adversaries gain access to the archived blocks, and choose to execute their attack on a subset of them based on underlying coding scheme. We show that our optimized threats can achieve the same level of damage as that of blind attacks, however, with limited amount of resources. Overall, this is the first work of its kind that opens up new questions on designing coded blockchains to jointly provide storage savings, scalability and also resilience to optimized threats.

\end{abstract}



\end{frontmatter}



\section{Introduction and Background}
\label{Introduction}

Blockchain is a decentralized ledger technology that works on the principle of distributed trust among the peers of the underlying network. Since blockchain is not governed by a single trusted third party, it avoids single point of failure. Furthermore, it offers several other advantages such as (i) open accessibility, (ii) tighter security by chaining the blocks using cryptographic hashes, (iii) data is tamper-proof, (iv) fault-tolerance as the information is archived by storing duplicate copies at a large number of trustworthy nodes, (v) ensuring transparency through distributed consensus algorithms, and (vi) immutability \citep{Block_Beyond_Bitcoin}. Consequently, blockchain has emerged as the underlying technology for cryptocurrencies such as bitcoin \citep{Bitcoin} and ethereum \citep{Blockchain}. Furthermore, it has gained popularity in a wide range of applications such as Internet of Things (IoT) \citep{IoT}, healthcare (to track, verify and secure health data)\citep{Healthcare}, supply chain management \citep{Supplychain}, digital rights management \citep{Digital_Rights} and in several other domains \citep{Blockchain_Applications}.

Blockchain's distinguished security properties including decentralization, distributed trust, immutability, and transparency are safeguarded by a special class of fully functional nodes, known as \emph{full nodes}. These nodes serve as the backbone of a blockchain network by replicating and storing all blocks. The key responsibilities of a full node include archiving blockchain history, independently validating new blocks and bootstrapping new nodes joining the network. As a result, full nodes require a large storage space. For instance, the size of the bitcoin blockchain has grown over $554\emph{GB}$ as of March 2024 \citep{Blockchain Size}.
Similarly, the XRP ledger size is approximately $26\emph{TB}$ as of July 2023, in which the storage requirement per full node increases by approximately $12\emph{GB}$ per day \citep{XRPSize}. As a consequence, scalability of blockchains is limited by its high storage requirements. This causes a significant drop in the number of full nodes, which violates blockchain's decentralization feature. For instance, the number of operational nodes of Bitcoin blockchain fell from $200,000$ in $2018$ to $47000$ in $2020$ \citep{Bitcoin node count falls}. These numbers clearly indicate that storage is a bottleneck for full nodes in blockchains. Hence, much research effort has been dedicated towards addressing the problem of high storage requirement in blockchains.

In the recent past, several approaches such as simplified payment verification \citep[Section 8]{Bitcoin}, block pruning \citep{Pruning} and sharding \citep{Sharding} have been proposed to cut down the storage cost of full nodes. However, these approaches are prone to security issues, as reported in \citep{Survey paper}. 
As a result, coded-blockchain \citep{Coded Blockchain} has emerged as an alternative approach to securely scale the blockchain with reduced storage requirement on the devices of the nodes. 

The idea of coded blockchain originates from distributed storage settings wherein all the archived blocks of a node, say a full node, are erasure coded using an underlying coding strategy to generate coded blocks. Subsequently, these coded blocks are distributed across a number of storage devices, referred to as droplet nodes, thereby reducing storage requirement per device and also enhancing reliability against device failures \citep{Erasure_code}. With such an architecture, a new node that wishes to mirror the blockchain, would need to contact sufficient number of droplet nodes, download their data and then apply a decoding strategy to recover the archived-blocks. While coded blockchains enjoy the benefits in reduced storage size, it requires new nodes to download more data compared to the uncoded counterparts. In this line, \citep{SeF} proposed a secure fountain (SeF) architecture for blockchains by using Luby Transform (LT) codes \citep{LT_codes} to provide storage benefits.
Furthermore, LT codes were chosen in \citep{SeF} as they offer the following advantages on scalability: (i) New nodes would need to download marginally larger-sized data than the uncoded counterpart (ii) New nodes have the flexibility to use a wide range of low-complexity LT decoders depending on their computational complexity and communication overhead and (iii) LT coded blockchain systems are resilient against node failures. 

In this context, we address two major challenges associated with LT coded blockchains, owing to which this paper is divided into two parts. Initially, we work on the problem of designing optimal decoders that facilitate seamless scaling of LT coded blockchains under heterogeneous network constraints. To this end, we present new decoders for LT coded blockchains, that are customized to heterogeneous nodes with varying download and computing capabilities, so as to facilitate low cost mirroring of the blockchain. In the second part, we study the vulnerabilities of the LT coded structure that may be exploited by an adversary to deny a certain group of new nodes from decoding the original blockchain and joining the network, thereby stalling its scalability.

Towards designing low complexity decoders for LT coded blockchains, the prior works have aimed to either reduce the computational complexity by compromising on the bootstrap overhead using the Belief Propagation (BP) decoder \citep{SeF,SRED}, or vice-versa using the On the Fly Gaussian Elimination (OFG) decoder \citep{OFG}. However, in a blockchain network consisting of nodes of heterogeneous download and computing capabilities, standalone BP and OFG decoders are not suitable as they may not optimally trade the bootstrap overhead with computational complexity, depending on the needs of a heterogeneous node. Furthermore, under such heterogeneous network constraints, it is important to take into account both complexity as well as overhead metrics simultaneously, while designing a decoder. To this end, we design advanced LT decoders for heterogeneous nodes, which facilitates low cost decoding for any network constraint.

Furthermore, we are also interested in analyzing the vulnerability of coded blockchain architectures against security threats that target their scalability feature. Since full nodes assist in bootstrapping new nodes in the network, it is essential to study their vulnerability on a wide range of active threats. One way in which an adversary or a group of adversaries could jeopardize scalability is to contact full nodes as potential new nodes and then deny them from serving other new nodes in the network by flooding service requests \citep{DoS_1,DoS_2}. Another option for the adversaries is to contact full nodes, compromise their access structure and then manipulate their coded blocks in the form of integrity threats or ransomware threats \citep{Ransomware_healthcare,Ransomware_CryptoLocker,Ransomware_Petya}. Consequently, such attacks forbid new nodes from downloading the contents of those full nodes, thereby stalling scalability. With such potential threats on the full nodes in blockchain, in this work, we explore whether the coded structure of LT-coded architecture such as \citep{SeF} provides any advantage to the adversaries in terms of the attack strategies.

\begin{figure*}[ht!]
    \centering
    \includegraphics[scale = 0.4]{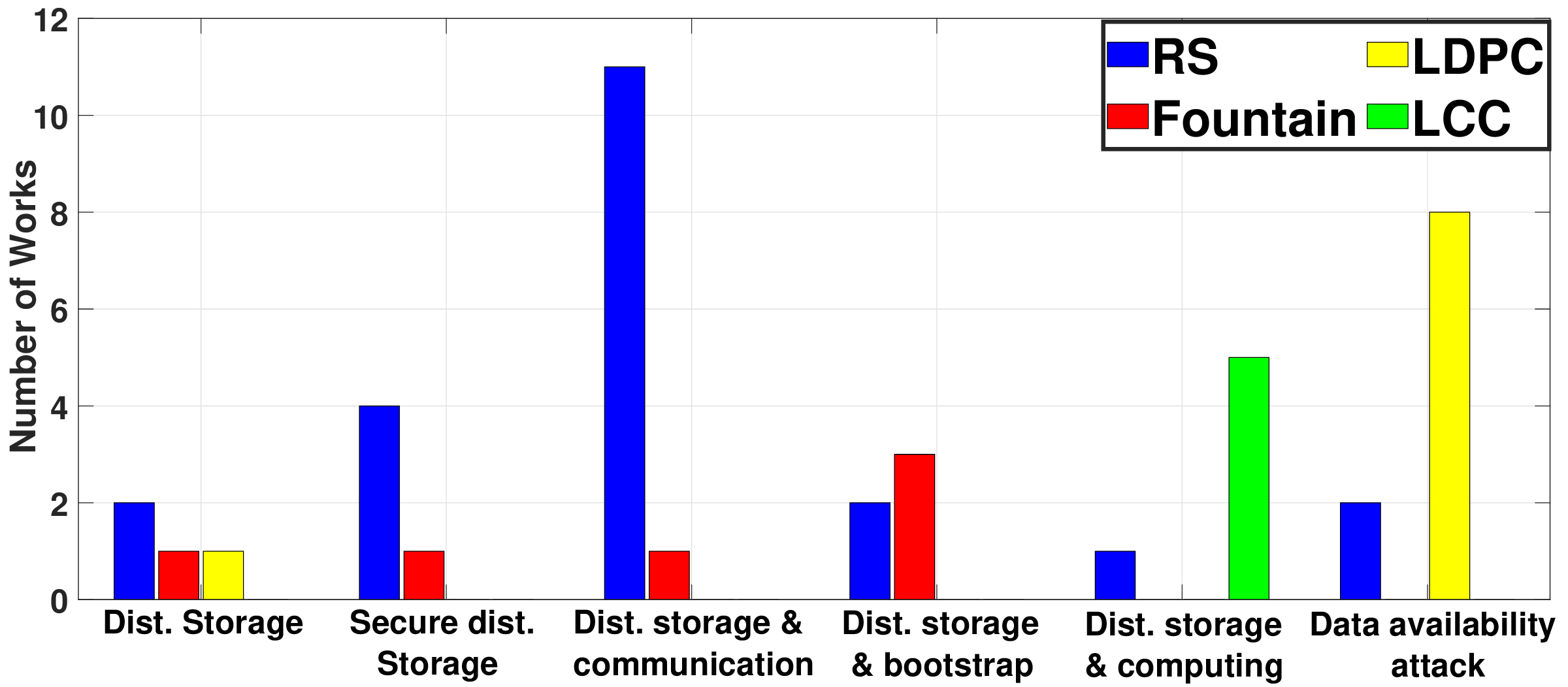}
    \caption{\textcolor{black}{Number of contributions that implement RS, Fountain, LDPC and LCC codes to address scalability issues in blockchains (source: \citep{Survey paper}).}}
    \label{fig:Related_work}
\end{figure*}

\subsection{Our Contributions}

Our key contributions towards scaling LT coded blockchains under heterogeneous network constraints are summarized as follows. 

\begin{enumerate}
    \item We introduce two novel variants of hybrid decoders for LT coded blockchains, namely: the Bootstrap-Rigid Hybrid (BRH) and the Complexity-Rigid Hybrid (CRH) decoders, both of which exploit the low complexity operations of the BP decoder before switching to their OFG counterpart in an opportunistic sense. As a result, both of them offer flexibility to a new node by allowing it to operate on its feasible range of computational complexity and bootstrap overhead.
    \item To equitably compare the traditional LT decoders with the newly introduced BRH and CRH decoders, we unify the costs associated with computation and download at the new node. Subsequently, we propose a new metric known as the \emph{mirroring cost}, which takes into account the heterogeneity of the nodes in the blockchain network as well as the complexity and overhead of the decoder used. We then formulate optimization problems for both BRH and CRH decoders to minimize their respective mirroring costs, taking into account the new node's download and computing capabilities.
    \item Through comprehensive experiments, we establish that for an LT coded blockchain node with given network constraints, BRH/CRH decoders operating at their optimal overhead/complexity combination allow decoding of the blockchain at a lower mirroring cost than that of the standalone BP and OFG decoders. Overall, our analysis provides pointers on how to choose the parameters of BRH and CRH decoders in practice as they provide joint benefits in both bootstrap overhead and computational complexity.
\end{enumerate}

 Next, we summarize our contributions on the vulnerabilities of LT coded blockchain architectures to DoS threats. 

\begin{enumerate}
    \item We present a novel threat model on LT coded blockchains, wherein an adversary can read a specific fraction of the available servers of a full node, and then use this information to launch DoS attacks on a subset of them, thereby forbidding new nodes from downloading the blockchain from the full nodes. By referring to the above class of adversaries as non-oblivious adversaries, we propose reasonable attack strategies corresponding to each of BP, OFG, BRH and CRH decoders. The proposed attack strategies use the information read by the adversary and then select a subset of servers on which DoS attacks has to be launched, so as to produce a high decoding failure rate for a new node joining the network. Through comprehensive experimental results, we show that the proposed non-oblivious adversaries achieve a higher level of damage than that of their oblivious counterparts blind attacks.  
    \item We also consider a cost constrained adversary, for which the total cost that the adversary could spend jointly in terms of reading the contents of the storage nodes as well as launching the DoS attack on a subset of them, is fixed. Under such a constraint, our analysis assists the adversary to decide on what fraction of the storage nodes has to be read/attacked to incur maximum decoding failure rate in the decoder he targets.  
\end{enumerate}


\textcolor{black}{The security threats proposed in this work inherently capture the vulnerabilities of the vanilla version of the LT coded blockchains \citep{SeF}, wherein decoders based on belief propagation (BP) and on-the-fly Gaussian (OFG) elimination are applicable. Therefore, the impact of our attacks are first presented against such an architecture. In addition, given that we also  propose new enhancements to the LT coded framework in the form of BRH and CRH decoders, we also extend the vulnerability study to those enhancements.} Overall, our work is the first of its kind in designing optimal decoders for heterogeneous network constraints and also analysing the threats caused by non-oblivious adversaries in LT coded blockchain architecture.

\subsection{Related Work}

Since coded blockchain addresses both scalability and security issues associated with blockchains, there are two major lines of research associated with it. Most of the works focuses on the scalability of blockchains, aiming to reduce the storage requirement of full nodes by using error correction codes. Another line of research studies the vulnerabilities of coded blockchains to Data Availability Attacks (DAA). The works associated with the aforementioned research lines are discussed in Section \ref{Section:Related_Work_Scalability} and Section \ref{Section:Related_Work_Security}.

\subsubsection{Blockchain Scalability using Error Correction Codes}\label{Section:Related_Work_Scalability}


In the context of enhancing blockchain scalability using error correction codes, researchers have predominantly explored the communication overhead, bootstrap costs, storage optimizations, and computational aspects of coded blockchain framework. The most common error correction codes used for the above-mentioned purposes are, Reed-Solomon (RS) code \citep{RS_codes}, Low Density Parity Check (LDPC) code \citep{LDPC}, Luby Transform (LT) code \citep{LT_codes}, Raptor code \citep{Raptor} and Lagrange Polynomial code \citep{Lagrange}. Fig. \ref{fig:Related_work} shows the number of works in which the aforementioned codes are employed to address various scalability issues in blockchain \citep{Survey paper}. It is clear from the statistics shown in Fig. \ref{fig:Related_work} that RS codes are widely employed to optimize the communication overhead, owing to its Maximum Distance Separable (MDS) property. However, it does not support low-complexity reconstruction operations due to higher-order finite field arithmetic. On the other hand, the class of fountain codes are predominantly used to optimally trade storage savings with bootstrap cost \citep{SeF,SRED}. It is also evident that LDPC codes are used to overcome DAA in blockchain systems \citep{LDPC_DAA}, whereas Lagrange Coded Computing (LCC) codes are widely used in applications where the devices operating the blockchain have limited computational capabilities \citep{LCC_IoT}.

\begin{figure*}[ht!]
\begin{center}
  \includegraphics[scale = 0.4]{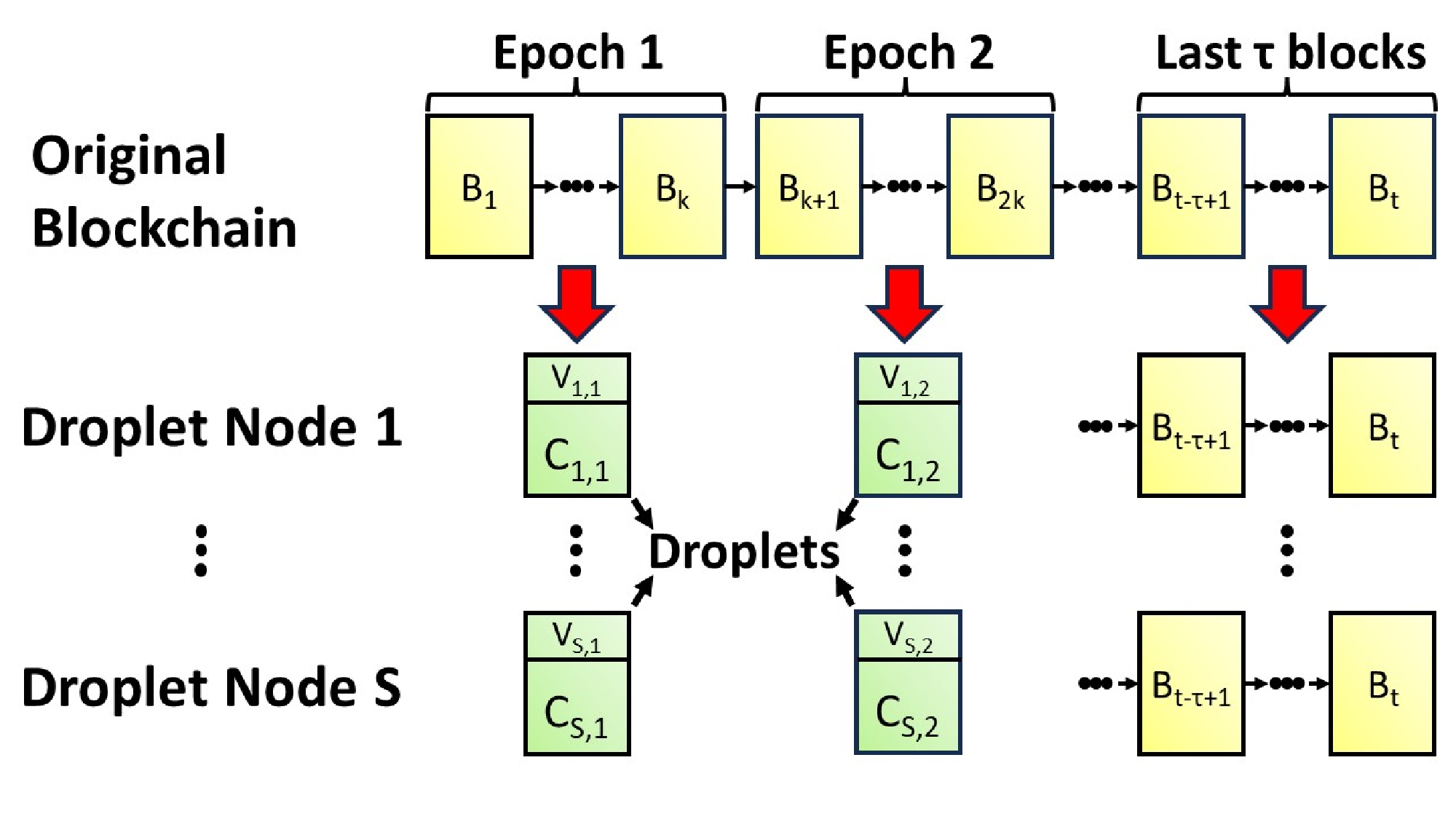}
  \end{center}
  \vspace{-1cm}
  \caption{Depiction of coded blockchain architecture wherein $k$ blocks of an epoch are coded to generate $S$ coded blocks and then stored across $S$ droplet nodes. The vectors $\mathbf{v}$ store the indices of the message blocks used to generate the corresponding coded block, thus capturing the degree information. The most recent $\tau$ blocks are stored in an uncoded format to handle blockchain reorganizations due to potential forks \citep{SeF}.}
  \label{fig:Storage_Model}
\end{figure*}

Furthermore, as another interesting direction of research, several coded architectures have been proposed to address the scalability issue in blockchains. In this context, authors of \citep{SeF} present a secure fountain (SeF) architecture for blockchains by using LT codes to provide storage benefits. Ultimately, \citep{SeF} aims to achieve a near-optimal trade-off between bootstrap overhead and storage savings. Authors of \citep{SRED} have also worked on similar lines as \citep{SeF}, and suggested raptor codes to be used in place of LT codes to achieve an almost constant decoding complexity, however, with a higher bootstrap overhead. Subsequently, \citep{WCNC_paper} pointed out the problem of high bandwidth requirement for decoding a particular requested block, faced by \citep{SeF,SRED}. To this end, \citep{WCNC_paper} presented two different architectures namely (i) Secure and Private Fountain codes, (ii) Secure and Private Randomly Sampled codes to obtain storage savings, low repair bandwidth per epoch, low communication cost for decoding a requested block, and simultaneously provide privacy of the archival data. However, unlike \citep{SeF,SRED}, the architectures in \citep{WCNC_paper} incur higher costs for decoding the complete chain. Another work \citep{Regenerating_codes}, presents a secure-repair-block protocol based on regenerating codes for sharded blockchains to reduce both storage and bootstrap costs, which makes it easier for a new miner to enter the system. A different study \citep{Homomorphic_hashes} introduces homomorphic hashes in coded blockchains, which allows instantaneous detection of erroneous fragments, thereby avoiding decoding with wrong data. \textcolor{black}{A survey on coded blockchains is also available in \citep{new_survey}}.

\textcolor{black}{Among the various error correction coding techniques, we adopt LT codes in this work to enhance the scalability, reliability, and data availability of blockchain systems. By fragmenting blockchain data into smaller, redundant coded symbols, LT codes enable efficient data recovery even when some fragments are lost or corrupted, reducing the per-device storage load and enhancing resilience to node failures or network disruptions \citep{BC_IoT}. Unlike traditional erasure codes, LT codes are non-systematic and do not rely on a fixed generator matrix; each coded block is equally informative, allowing decentralized, load-balanced retrieval from any subset of fragments \citep{Consortium}. This property is particularly beneficial in decentralized environments, as it prevents traffic bottlenecks when multiple nodes simultaneously attempt to sync with the blockchain. Furthermore, LT coding minimizes the need for retransmissions in unreliable networks by enabling decoding from any sufficiently large subset of coded symbols \citep{Facebook}, thereby improving transmission efficiency and supporting scalability. As decoding can begin once enough fragments are received, LT-coded systems can also accelerate transaction confirmations, making them well-suited for distributed networks with limited resources.}

\subsubsection{Data Availability Attacks in Blockchains} \label{Section:Related_Work_Security}

Typically, blockchains face security threats in the form of DAA. This is because, when the ratio of light nodes to full nodes increases due to rising storage requirements on full nodes, only a smaller number of nodes will have access to the entire set of data or transactions. As a result, some of the full nodes may engage in malicious behaviours, as the light nodes are unable to detect any malicious activity such as invalid or missing transactions \citep{DAA1,DAA2}. In recent times, several approaches have been proposed to address the DAA issue in blockchains. In one such study, coded merkle tree is proposed as a constant-cost protection against blockchain based DAA \citep{DAA3}. Another line of work uses 2-D RS code with merkle trees to mitigate DAA \citep{DAA4,DAA5}. However, 2-D RS code experiences large decoding complexity and coding fraud proof size when applied to blockchains with large blocks. To this end, \citep{DAA6} constructs a merkle tree using polar codes that can be used to mitigate DAA, with a lower coding fraud proof size compared to that of 2D-RS codes.

Furthermore, DAA in the form of Denial-of-Service (DoS) threats faced by blockchains has also been widely studied. In this line, \citep{DAA7} and \citep{DAA8} thoroughly investigate the DoS attacks in blockchains and solutions to mitigate them. A later study introduces Blockchain DoS (BDoS), an incentive-based DoS attack that targets proof-of-work cryptocurrencies \citep{DAA9}. It is claimed that BDoS targets the system's mechanism design by exploiting the reward structure, effectively discouraging miner participation.

Some preliminary results of this work are available in \citep{My_Paper}, which captures a part of our contributions on the vulnerabilities of LT coded blockchains to optimized DoS threats. In particular, this work builds upon the work in \citep{My_Paper} by providing a comprehensive analysis of LT decoders' performance metrics and exploring the trade-off between computational complexity and bootstrap overhead, which \citep{My_Paper} does not address. While \citep{My_Paper} briefly discusses hybrid decoders, this work optimizes them for blockchain decoding at minimal cost under heterogeneous network constraints. In vulnerability analysis, we introduce more advanced and powerful attack strategies and efficient implementation techniques. Specifically, we propose a more effective score-based attack strategy on the BP decoder and devise a computationally efficient algorithm for attacking the OFG decoder, improving upon the attacks presented in \citep{My_Paper}. \textcolor{black}{To the best of our knowledge, coded blockchains have not been used in live blockchain deployments, however, they have been validated on real Bitcoin trace data through academic prototypes \citep{BC_IoT} and studied on IoT networks \citep{Recent_paper}. While commercial readiness of such architectures are possibly a few years away, we believe our study on the security vulnerabilities is imperative before large-scale deployments as it helps in reducing the risks, prevent downtime, and build more trust and compliance.}

\subsection{Organization}
The remainder of this paper is organized as follows. In Section $2$, we discuss the storage model of the LT coded blockchain architecture, followed by its encoding and decoding methodologies. In Section $3$, we introduce the BRH and CRH decoders for low cost mirroring of LT coded blockchains and carry out a detailed performance analysis of the various LT decoders. Furthermore, we optimize the newly introduced BRH and CRH decoders to mirror the blockchain at the lowest cost, for any heterogeneous constraints on the mirroring node. In Section $4$, we present reasonable attack strategies against various LT decoders along with feasible algorithms to execute them, and some experimental results supporting the effectiveness of the proposed strategies. Furthermore, we optimize the attacks from an adversary's perspective in order to maximize the decoding failure rate of the targeted decoder, subject to the adversary's cost constraint on the attack execution. Finally, a concluding summary and directions for further research are provided in Section $5$.

\vspace{-2mm}
\section{LT Coded Blockchain Architecture}
\label{sec2_model}

To accommodate the rapid growth in the size of a full node, coded blockchain architecture distributes the contents of a full node as coded fragments across multiple smaller-sized storage devices. Fig. \ref{fig:Storage_Model} shows a visual illustration of coded blockchain architecture. In this architecture, the original blockchain stored on a full node is divided into several \emph{epochs}, wherein an epoch is defined as a collection of $k$ blocks of the blockchain, for some $k \in \mathbb{Z}^{+}$. For every epoch, a sufficient number of coded blocks referred to as \emph{droplets} is generated, and these droplets are stored across storage constrained nodes referred to as \emph{droplet nodes}, as seen in Fig. \ref{fig:Storage_Model}. Henceforth, we assume that a droplet node will store one droplet per epoch, thus storing as many droplets as the number of epochs, as shown clearly in Fig. \ref{fig:Storage_Model}. This way, the storage size of a full node is slashed by distributing it across smaller-sized storage devices with roughly $\frac{1}{k}$ storage capacity. 
In the next section, we explain how a blockchain is encoded to obtain the droplet nodes using an LT coded architecture \citep{SeF}. 




\subsection{Slashing Storage Costs through LT Encoding}
\label{sec:Encoding}


When generating the droplets of a full node, LT encoding is done as follows. First, a number $d \in [k]$ is chosen randomly from a suitable probability distribution. Subsequently, $d$ distinct input blocks are uniformly chosen from the $k$ blocks of an epoch, and a bit-wise XOR of the chosen blocks, denoted by $C$, is obtained. Henceforth, this collection of $d$ input blocks of an epoch that are XORed to obtain $C$ is referred to as the neighbours of the droplet $C$. Along with $C$, a binary vector $\mathbf{v}$ of length $k$ is also generated, wherein the $b$-th entry of $\mathbf{v}$ is $1$ if the $b$-th input block is among the $d$ chosen blocks, else the $b$-th entry is $0$. This combination of $C$ and $\mathbf{v}$ constitute a droplet for an epoch. An illustration of this procedure is given in Fig. \ref{fig:LT_Encoding}, for the case when $k=6$. 

\begin{figure}
  \includegraphics[height = 6.5cm, width=\linewidth]{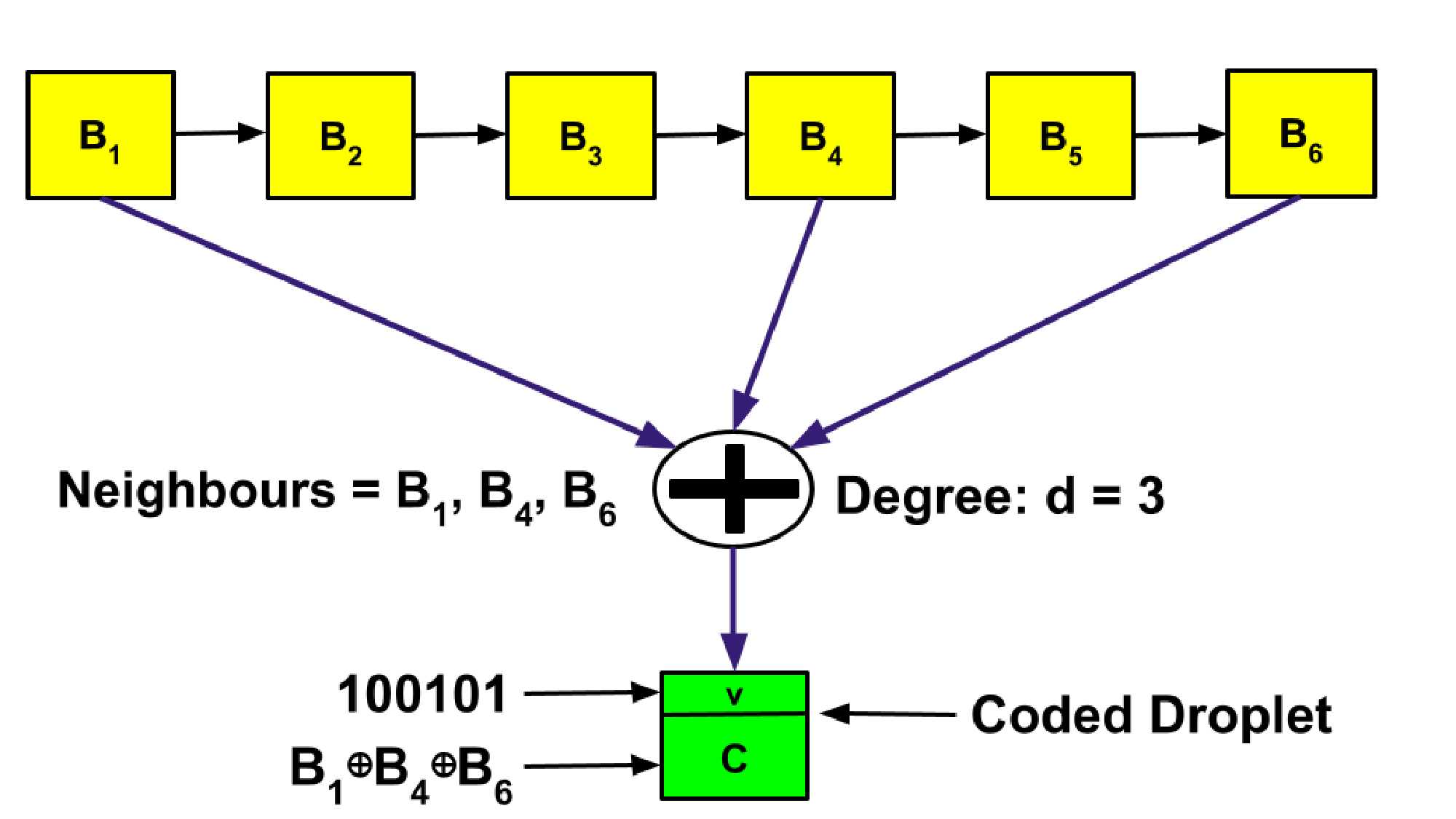}
  \caption{Depiction of generating a coded droplet from an epoch using LT encoding.}
  \label{fig:LT_Encoding}
\end{figure}

Furthermore, we assume that all the droplets of a droplet node that correspond to different epochs, will have the same degree and also the same set of neighbours.
In the LT code terminology, the number $d$, which is referred to as the degree, is picked from the popular Robust Soliton Distribution (RSD) \citep[Definition 11]{LT_codes}, denoted by $\Omega_{RS}(\cdot)$ \textcolor{black}{which is defined for a given set of parameters $0<\delta<1$ and $c>0$, as}
\begin{equation}
    \Omega_{RS}(d) = \frac{\rho(d)+\tau(d)}{\beta} \quad \forall d = 1,\hdots,k,
\end{equation}
where $\rho(\cdot)$ is the Ideal Soliton distribution \citep[Definition 9]{LT_codes} \textcolor{black}{defined as}
\begin{equation}
\rho(d)= 
\begin{cases}
   1/k & \text{for } d = 1\\
   \frac{1}{d(d-1)}  & \text{for } d = 2,\hdots,k ,\\ 
\end{cases}
\end{equation}
and $\tau(.)$ \textcolor{black}{is given by}

\begin{equation}
\tau(d)= 
\begin{cases}
   \frac{R}{dk} & \text{for } d = 1,\hdots,\frac{k}{R}-1\\
   \frac{R}{k}\cdot\ln{\left(\frac{R}{\delta}\right)}  & \text{for } d = \frac{k}{R}\\
   0 & \text{for } d = \frac{k}{R}+1,\hdots,k , 
\end{cases}
\end{equation}
in which $R = c\sqrt{k}\ln{\left(\frac{k}{\delta}\right)}$ and $\beta = \sum_{j=1}^{k} \rho(j)+\tau(j)$.

\vspace{5mm}

Using the above distribution, $d$ is chosen in a statistically independent manner and a total of $S$ droplet nodes, for $S > k$, are created at the full node along the similar lines.

\subsection{Blockchain Retrieval through Traditional LT Decoders}

A new storage-constrained node that wants to retrieve the blockchain from a full node is termed as a \emph{bucket node}. For the blockchain retrieval, the bucket node contacts a little more than $k$ droplet nodes (from the set of $S$ nodes), and downloads their droplets corresponding to all the epochs. To perform LT decoding, the bucket node can employ any one of the following traditional LT decoders, namely the BP, OFG and variants of the so-called hybrid decoders.

\begin{figure}
  \includegraphics[height = 6.5cm, width=\linewidth]{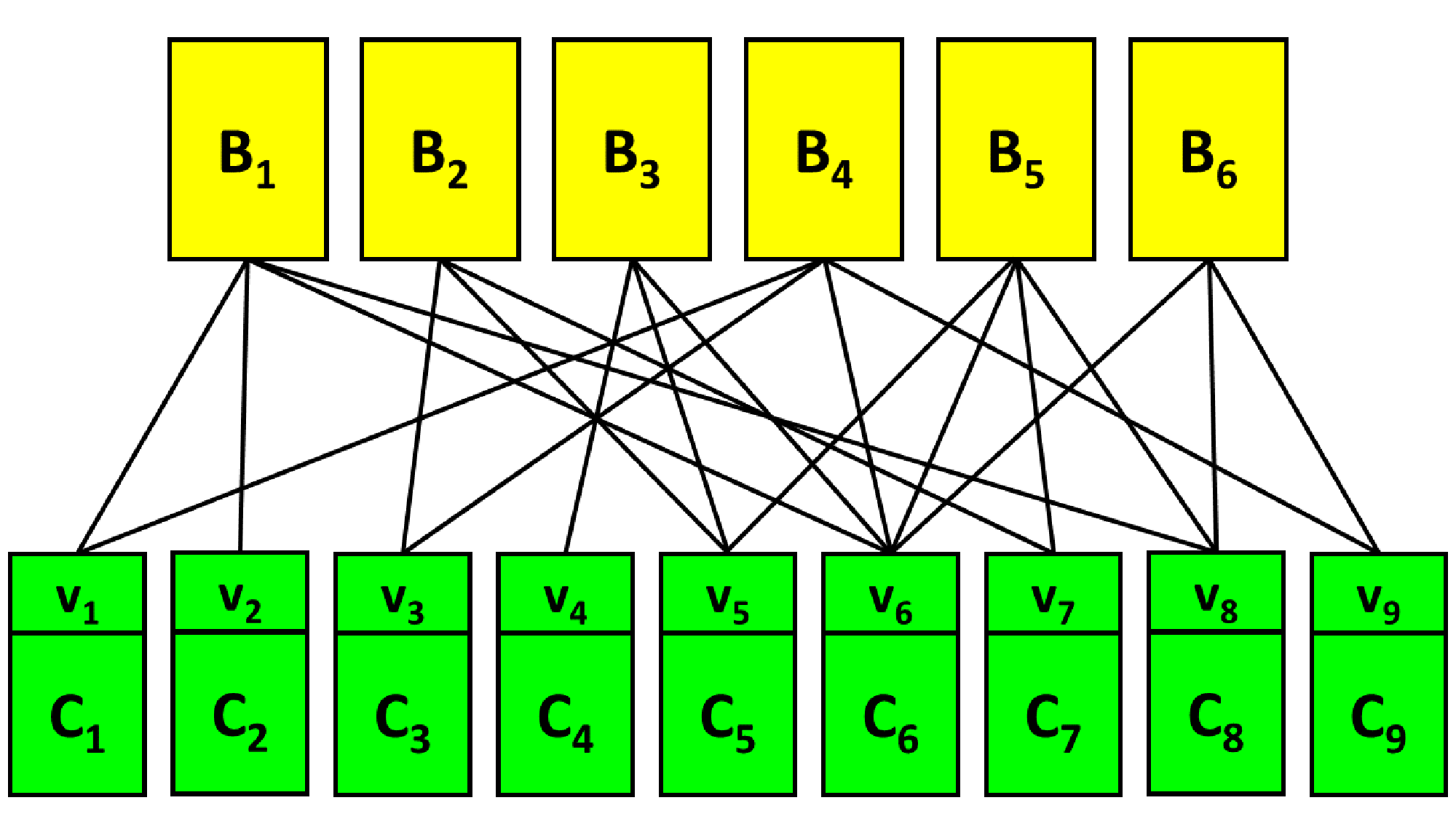}
  \caption{Illustration of bipartite graph $\mathcal{T}$ formed using the coded droplets collected by a bucket node.}
  \label{fig:Bipartite_Graph}
\end{figure}

\subsubsection{Belief Propagation (BP) Decoder}

The bucket node randomly contacts $K \ge k$ droplet nodes, for some $K \in \mathbb{N}$, and downloads their droplets. Since all the droplets corresponding to a particular droplet node share the same set of neighbours, decoding procedure of one epoch explains the decoding procedure of the entire blockchain. The decoding procedure to recover an epoch is explained as follows:
\begin{enumerate}
\item The bucket node forms a bipartite graph $\mathcal{T}$ with the $k$ input blocks as top vertices and the $K$ droplets as bottom vertices. The input blocks are denoted by $\{B_{m}~|~ m \in [k]\}$, and the droplets are denoted by $\{C_{i}~|~ i \in [K]\}$. An edge connects a droplet $C_{i}$ to an input block $B_{m}$ if $B_{m}$ is used in computing $C_{i}$. Note that this information is available from the vector $\mathbf{v}_{i}$ associated with each $C_{i}$. This step is illustrated in Fig. \ref{fig:Bipartite_Graph}.
\item The bucket node finds a droplet $C_{i}$, for $i \in [K]$, that is connected to exactly one input block $B_{m}$ in $\mathcal{T}$. Such a droplet is called a singleton. If there are no singletons, the bucket node declares a decoding failure and terminates the process.
\item If $C_{i}$ is a singleton, the bucket node sets $\hat B_{m} = C_{i}$, where $\hat B_{m}$ denotes the $m$-th decoded block. For all the droplets $C_{i}$ connected to the decoded block $B_{m}$ in $\mathcal{T}$, it sets $C_{i} \leftarrow C_{i} \oplus B_{m}$ ($\oplus$ denotes bit-wise XOR) and modifies $\mathcal{T}$ by removing all the edges connected to block $B_{m}$.
\item If all the $k$ blocks are not recovered, the bucket node navigates to Step 2 to repeat the procedure.
\end{enumerate}
For more details on the BP decoding process, we refer the reader to \citep[Section 3.2.3]{SeF}. In our work, the BP decoder is said to experience decoding failure if all the $k$ blocks of an epoch cannot be recovered using the $K$ droplet nodes.


\subsubsection{On-the-Fly Gaussian Elimination (OFG) Decoder} \label{sec:OFG decoder}

The OFG decoder prepares a binary generator matrix $\mathbf{G} \in \{0,1\}^{K \times k}$, whose rows are the binary vectors $\{\mathbf{v}_{i}~|~ i \in [K]\}$ stored along with the droplets $\{C_{i}~|~ i \in [K]\}$. Subsequently, the decoder works on the principle of Gaussian Elimination (GE) decoding, which requires $\mathbf{G}$ to have rank $k$ for decoding the $k$ blocks of an epoch. The main idea here is to use $\mathbf{G}$ to obtain a sparse upper triangular matrix by deleting redundant equations on the fly. Once the sparse upper triangular matrix is ready, back-substitution can be performed to recover the epoch. For more details on the OFG process, we refer the reader to \citep[Algorithm 1]{OFG}. Note that the OFG decoder is said to experience decoding failure if the rank of $\mathbf{G}$ is less than $k$.

Other than the BP and OFG decoders, LT coded blockchains can also be retrieved through a class of hybrid decoders.


\subsubsection{Hybrid Decoder} \label{sec:Hybrid decoder}
 Similar to the BP and OFG decoders, the bucket node contacts $K \ge k$ droplet nodes and downloads the droplets stored in them. The idea of hybrid decoder, as described in \citep[Section \Romannum{3}]{Hybrid_decoding_idea}, is to start with the BP decoding process using the $K$ downloaded droplets and retrieve the input blocks of an epoch one by one. Once the BP part of the hybrid decoder exhausts its singletons before retrieving all the $k$ input blocks, the non-singletons remaining in the BP part are used by its OFG counterpart to retrieve the remaining blocks. This idea exploits the advantage of low-complexity decoding of BP decoder by retrieving as many blocks as possible using BP, and also exploits the advantage of the OFG decoder, which guarantees successful decoding by contacting fewer droplet nodes than the BP decoder. Note that the hybrid decoder is said to experience decoding failure if it is unable to recover all $k$ blocks of an epoch from the $K$ droplet nodes despite using its BP and OFG counterparts.

In all of the preceding discussions, the primary goal was to explain the decoding procedures of the various LT decoders for a given value of $K$. In particular, we do not consider the case of incremental decoding, in which the bucket node contacts additional droplet nodes if decoding fails after utilizing the first $K$ droplet nodes. Overall, we emphasize that a bucket node can employ one of the above variants of the decoders depending on its constraints on the computational complexity, the number of droplet nodes it can contact and its desired failure rate. In the next section, we present new variants of the hybrid decoder, and discuss their benefits over the traditional BP and the OFG decoders. 

\section{Advanced Hybrid Decoders and their Performance Measures}
Towards proposing specific variants of the hybrid decoders, we define the following two main metrics, namely: \emph{bootstrap overhead} and \emph{computational complexity}. 

\begin{definition}[Bootstrap overhead]
    Bootstrap overhead is defined as the number of droplet nodes that a bucket node must contact in order to successfully decode an epoch.
\end{definition}

\begin{definition}[Computational complexity]
    Computational complexity is defined as the number of XOR operations performed between droplets to decode an epoch.
\end{definition}

It is important to note that for a bucket node employing any of the aforementioned decoders, the necessary condition for successful decoding is that the rank of the binary generator matrix $\mathbf{G}$ obtained using the coefficients of the $K$ droplets must to be equal to $k$. According to the working principle of the OFG decoder, a bucket node can successfully decode the blockchain if this condition is met. On the other hand, the BP decoder declares success only when it iteratively finds $k$ singletons during its decoding process. Consequently, for a given failure rate, a bucket node using the BP decoder needs to contact more droplet nodes compared to using the OFG decoder. However, the OFG decoding process relies on Gaussian Elimination, which involves matrix inversion to solve a set of linear equations, making it computationally intensive. In contrast, the BP decoder avoids such heavy computations, making it more efficient in terms of computational complexity. Therefore, there is a clear trade-off between the bootstrap overhead and computational complexity for the BP and OFG decoders. 
Hence, to optimally trade the bootstrap overhead with the computational complexity for bucket nodes, we propose the following two new variants of hybrid decoders.



\subsection{Bootstrap-Rigid Hybrid (BRH) Decoder}
With this decoder, the bucket node initially fixes on the number of droplet nodes that it contacts and it is unwilling to download from more droplet nodes. Let that fixed number of droplet nodes be $K$, for some $K > k$. Then, it executes the hybrid decoding method as explained in Section \ref{sec:Hybrid decoder}. Note that the BRH decoder focuses solely on the number of droplet nodes it contacts and does not consider what fraction of the epoch is being recovered by its BP and OFG counterparts. Consequently, the BRH decoder is rigid in terms of bootstrap overhead and flexible in terms of computational complexity. It is worth mentioning that the OFG decoder can be a special case of BRH decoder if the BRH decoder decides to recover the epoch exclusively using its OFG counterpart. This situation arises when the BRH decoder cannot find a singleton among the $K$ collected droplets. Given the decoder's flexibility in computational complexity, it will experience decoding failure only if the rank of $\mathbf{G}$ is less than $k$, where $\mathbf{G}\in \{0,1\}^{K \times k}$ is the binary generator matrix as explained in Section \ref{sec:OFG decoder}.

\subsection{Complexity-Rigid Hybrid (CRH) Decoder}

Unlike the BRH decoder, the CRH decoder can contact up to $S$ droplet nodes, where $S$ is the total number of droplet nodes available at a full node. On the other hand, a bucket node employing the CRH decoder fixes on the minimum number of blocks of an epoch that must be retrieved exclusively by its BP component in order to restrict its computational complexity. Let that minimum number be denoted by $\eta_{c}$, such that $0\le\eta_{c}\le k$. Then the bucket node first contacts $K_{init}$ droplet nodes, where $k<K_{init}<S$, and performs the hybrid decoding technique described in Section \ref{sec:Hybrid decoder}. Once the target $\eta_{c}$ is met using $K_{init}$ droplet nodes, the rest of the blocks can be decoded either using the BP component or the OFG component. It is important to note that the bucket node iteratively contacts more droplet nodes in a sequential manner and then conducts the hybrid decoding technique, in either of the following cases: (i) when a minimum of $\eta_{c}$ blocks cannot be decoded by the BP component, (ii) when the residual $k - \eta_{c}$ blocks cannot be decoded by the OFG component after successfully decoding the first $\eta_{c}$ blocks using the BP component. If either of the aforementioned scenarios occurs despite the bucket node utilizing all $S$ droplet nodes for decoding, the CRH decoder will experience a decoding failure.

\subsection{Performance Analysis}

Now that we have defined the operations of various LT decoders, in this section we derive their respective bootstrap overheads and computational complexities. We then use these metrics to compare their performances. 

\begin{proposition}\thlabel{Overhead_BP}
The average bootstrap overhead of a bucket node employing the BP decoder, to successfully decode an epoch with a probability greater than $1-\delta$ is given by $K_{BP} = k + c\sqrt{k}\ln^2{\left(\frac{k}{\delta}\right)}$, \textcolor{black}{where $c>0$ and $0<\delta<1$ are parameters of RSD.}
\end{proposition}

\begin{proof}
The bootstrap overhead of a bucket node is equal to the number of droplets needed to decode an epoch, as each droplet node stores one droplet per epoch and all the droplets of a droplet node corresponding to different epochs have the same degree and also the same selection of indexes of the input blocks.
Therefore from \citep[Theorem 12]{LT_codes}, the average bootstrap overhead is given by 
\begin{equation} \label{eqn:Upper bound bootstrap cost BP}
K_{BP} \le k + R\cdot H\left(\frac{k}{R}\right) + R\cdot \ln{\left(\frac{R}{\delta}\right)},
\end{equation}
where $H(n)$ is the Harmonic number given by $H(n) = \sum_{k=1}^{n} \frac{1}{k}$.
Now, we approximate $H\left(\frac{k}{R}\right)$ to $\ln{\left(\frac{k}{R}\right)}$ and substitute the value of $R = c\sqrt{k}\ln{\left(\frac{k}{\delta}\right)}$ in \eqref{eqn:Upper bound bootstrap cost BP} to get the average bootstrap overhead.
\end{proof}

\begin{proposition}\thlabel{Computation_BP}
The average computational complexity incurred by the BP decoder for decoding an epoch is given by $C_{BP} = k\cdot \left(\ln\left(\frac{c\cdot k^\frac{3}{2}\cdot \ln(k/\delta)}{\delta}\right) + 1.577\right)$.
\end{proposition}

\begin{proof}
The computational complexity for decoding an epoch can be determined by the number of edges in the bipartite graph $\mathcal{T}$, generated at the start of the BP decoder \citep[Appendix B]{SeF}. This is because, each XOR operation between two droplets corresponds to removal of an edge in the bipartite graph. Subsequently, all the edges of the bipartite graph will be removed at the end of BP decoding process, thereby resulting in as many number of XOR operations. It is important to note that the number of edges in the bipartite graph is random in nature, by the virtue of the random selection of the droplet nodes. Therefore, we calculate its mean value to determine the average computational complexity. In this regard, the average number of edges can be approximated as the product of average bootstrap overhead and the average degree of a droplet. In particular, if $D$ is the average degree of a droplet, \textcolor{black}{according to \citep[Theorem 13]{LT_codes}, we have:}

\textcolor{black}{\begin{equation}\label{Average_degree}
D = \frac{\sum_{i=1}^{k} i\cdot(\rho(i)+\tau(i))}{\beta}.
\end{equation}}
\textcolor{black}{Recalling the expressions for $\rho(i)$, $\tau(i)$, and $\beta$ as given in Section \ref{sec:Encoding}, and substituting them into \eqref{Average_degree}, we obtain}

\textcolor{black}{\begin{equation}\label{Degree_simplified}
    D = \frac{1}{\beta}\left(\frac{1}{k}+\sum_{i=2}^{k}\frac{1}{i-1}+\sum_{i=1}^{\frac{k}{R}-1}\frac{R}{k}+\ln\left(\frac{R}{\delta}\right)\right).
\end{equation}}

\textcolor{black}{Further, \eqref{Degree_simplified} can be simplified as follows}

\textcolor{black}{\begin{equation}\label{Degree_simplified_2}
    D = \frac{1}{\beta}\left(\sum_{i=2}^{k+1}\frac{1}{i-1}+\sum_{i=1}^{\frac{k}{R}-1}\frac{R}{k}+\ln\left(\frac{R}{\delta}\right)\right).
\end{equation}}

\textcolor{black}{Finally, \eqref{Degree_simplified_2} can be further approximated as follows}
\textcolor{black}{\begin{equation}\label{Degree_approximated}
    D \approx \frac{1}{\beta}\left(H(k)+1+\ln\left(\frac{R}{\delta}\right)\right),
\end{equation}}
\noindent \textcolor{black}{where $H(k)$ is the Harmonic number given by $H(k) = \sum_{n=1}^{k} \frac{1}{n}$. According to \citep{Harmonic_number}, for large $k$, the harmonic number can be approximated as $H(k) \approx \ln(k)+\gamma$, where $\gamma \approx 0.577$ is the Euler–Mascheroni constant. Now, substituting $R = c\sqrt{k}\ln{\left(\frac{k}{\delta}\right)}$ and $H(k) = \ln(k)+0.577$ in \eqref{Degree_approximated}, we obtain the average degree of a droplet as:}

\textcolor{black}{\begin{equation} \label{eqn:D}
D = \frac{1}{\beta}\cdot \left(\ln\left(\frac{c\cdot k^\frac{3}{2}\cdot \ln(k/\delta)}{\delta}\right) + 1.577\right).
\end{equation}}

\textcolor{black}{According to \citep[Theorem 12]{LT_codes}, $K_{BP}=k\beta$. With this, we arrive at the average computational complexity as $C_{BP} = K_{BP}D$.}

\end{proof}

The above two results present closed-form expressions on the average bootstrap overhead and average computational complexity of the BP decoder. \textcolor{black}{Given that the performance of the BP decoder is closely tied with the parameters of the RSD, we observe from \eqref{eqn:Upper bound bootstrap cost BP} and Proposition \ref{Computation_BP} that the bootstrap overhead and the average computational complexity of the BP decoder are inversely proportional to the allowable decoding failure probability $\delta$. Along similar lines, we also observe that both these metrics are directly proportional to the RSD constant $c$.}

Now, to derive the bootstrap overhead of OFG decoder, we must investigate the relationship between its probability of decoding failure and the number of downloaded droplets per epoch. To this end, authors of \citep{upper_bound} have established an upper bound on the decoding failure probability, denoted by ${P_{F}}$, for an LT code whose degree follow any arbitrary degree distribution. As a result, for binary LT codes whose degrees follow RSD, the upper bound in \citep[Theorem 1]{upper_bound} can be modified as follows:

\vspace{-3mm}
\begin{flalign} \label{eq:Upper bound of data erasure probability}
{P_{F}} &=  \sum_{w=1}^{k}{k \choose w}&&\\\nonumber
&\cdot\left[\frac{1}{2}\textcolor{black}{\sum_{d=1}^{k}}\Omega_{RS}(d)\frac{\sum_{l=0}^{d}{w \choose l}{k-w \choose d-l}\left[1-(-1)^{1-l}\right]}{{k \choose d}}\right]^{k(1+\epsilon)},&&
\end{flalign}
where \textcolor{black}{$\epsilon>0$} is the fractional overhead, defined as the ratio of excess droplets (in addition to $k$) and $k$. With that, an upper bound on the bootstrap overhead of the OFG decoder is given in \thref{Overhead_OFG}.

\begin{theorem}\thlabel{Overhead_OFG}
Let $\epsilon_{min}$ = $\frac{m}{k}$ where $m\in\mathbb{N}$ such that $P_{F}|_{\epsilon=\epsilon_{min}}\le\delta$, then the bootstrap overhead of OFG decoder is upper bounded by $K_{OFG} = k(1+\epsilon_{min})$.
\end{theorem}

\begin{proof}
From \eqref{eq:Upper bound of data erasure probability}, ${P_{F}}$ is clearly a decreasing function of $\epsilon$. The minimum value of $\epsilon$ (denoted by $\epsilon_{min}$) that satisfies the condition ${P_{F}} \leq \delta$, gives the minimum fractional overhead for an allowable probability of decoding failure $\delta$. Therefore, $k(1+\epsilon_{min})$ is an upper bound on the bootstrap overhead of the OFG decoder.   
\end{proof}

\begin{theorem}\thlabel{Computation_OFG}
Computational complexity of OFG decoder, denoted by $C_{OFG}$ is upper bounded by $k^2$. 
\end{theorem}

\begin{proof}
Computational complexity of the OFG decoder is evaluated as the sum of number of XOR operations performed by the OFG algorithm and back substitution, which together completes the decoding process. In the OFG algorithm, the maximum number of XOR operations associated with the $i^{th}$ droplet, whose binary coefficient entering the partially filled $k\times k$ matrix, is $i$. Hence, to fill the $k$ rows of the matrix, the number of XOR operations performed by the decoder is upper bounded by $\sum_{i=1}^{k} i = \frac{k(k+1)}{2}$. 
In the back substitution method, the number of XOR operations performed by the decoder at the $i^{th}$ iteration is $k-i$, wherein $i \in [k]$. Therefore, the total number of XOR operations performed at this stage is $\sum_{i=1}^{k} (k-i) = \frac{k(k-1)}{2}$. Adding the number of XOR operations associated with the OFG algorithm as well as the back substitution, we get an upper bound on the computational complexity as $C_{OFG}=k^{2}$.
\end{proof}

The above two results present the closed form expressions on the bootstrap overhead and computational complexity of the OFG decoder. 

Next, for the BRH decoder, it is clear that its bootstrap overhead, denoted by $K_{BR}$, is in the range $K_{OFG} \le K_{BR} \le K_{BP}$. Furthermore, it is also clear that its decoding failure rate is upper bounded by that of the OFG decoder, as $K_{BR}$ is atleast $K_{OFG}$. 
When implementing the BRH decoder, choosing values of $K_{BR}$ closer to $K_{OFG}$ leads to achieve a low bootstrap overhead, however, it compromises on computational complexity. The behaviour is vice-versa for choosing $K_{BR}$ closer to $K_{BP}$. Let $\eta_{B}$ be the number of input blocks of an epoch recovered by the BP component of the BRH decoder using $K_{BR}$ droplets. Since $\eta_{B}$ is a discrete random variable that takes values from $0$ to $k$, we can numerically evaluate the average value of $\eta_{B}$, denoted by $\EX[\eta_{B}]$ for each value of $K_{BR}$. The remaining $k-\eta_{B}$ blocks will be recovered using the OFG component. Hence, the average computational complexity of the OFG component of the BRH decoder, denoted by $\EX[(k-\eta_{B})^2]$ can also be evaluated for each value of $K_{BR}$. With that, the average computational complexity of BRH decoder is given below.

\begin{theorem}\thlabel{Computation_hybrid}
The average computational complexity of the BRH decoder is given by $C_{BR} = K_{BR}\cdot D\cdot \frac{\EX[\eta_{B}]}{k} + \EX[(k-\eta_{B})^2]$, \textcolor{black}{in which $K_{OFG}\le K_{BR}\le K_{BP}$ and $\eta_B \in \{0, 1, 2, \dots, k\}, \ \text{where } k \in \mathbb{Z}^+$.} 
\end{theorem}

\begin{proof}
The average computational complexity of BP part of the BRH decoder can be regarded as the average number of edges of the bipartite graph utilised by the decoder for recovering $\EX[\eta_{B}]$ blocks. From the complexity results on the BP decoder, we know that the average number of edges of the bipartite graph in the beginning of decoding is $K_{BR} D$. Furthermore, with the assumption that the number of edges originating from each input symbol is identical owing to uniform selection of input symbols in the encoding process, the average number of edges utilised by the BP component of the BRH decoder is $K_{BR} \cdot D \cdot \frac{\EX[\eta_{B}]}{k}$.
Furthermore, we have the average computational complexity of the OFG component of the BRH decoder as $\EX[(k-\eta_{B})^2]$ from the previous result. 
The sum of these two terms gives the average computational complexity of the BRH decoder. 
\end{proof}

 Next, we derive the metrics for the CRH decoder. Recall that the minimum number of input blocks of an epoch that must be recovered exclusively by the BP component of the CRH decoder denoted by $\eta_{c}$, is in the range $0\le \eta_{c} \le k$. Note that the bootstrap overhead of the CRH decoder is directly dependent on the $\eta_{c}$ set by the bucket node. When implementing the CRH decoder, it is intuitive that choosing values of $\eta_{c}$ closer to $k$ leads to a lower computational complexity, however it compromises on the bootstrap overhead. In contrast, the behaviour is vice-versa for choosing $\eta_{c}$ closer to $0$. Furthermore, it is not straight-forward to determine the exact bootstrap overhead of the CRH decoder, for a given $\eta_{c}$. Nevertheless, for each value of $\eta_{c} \in [k]$, we can numerically evaluate the average bootstrap overhead denoted by $K_{CR}^{av}$, through simulations. As $\eta_{c}$ is fixed, the computational complexity of the OFG component of the CRH decoder will be $(k-\eta_{c})^{2}$, which is a constant. With that, we present the average computational complexity of the CRH decoder as follows. 

\begin{theorem}\thlabel{Computation_CRH}
The average computational complexity of the CRH decoder is given by $C_{CR} = K_{CR}^{av}\cdot D\cdot \frac{\eta_{c}}{k} + (k-\eta_{c})^2$, \textcolor{black}{where $0\le \eta_c\le k$.} 
\end{theorem}

\begin{proof}
Explanation follows from the proof of \thref{Computation_hybrid}.
\end{proof}

\subsection{Complexity vs Overhead}
In Fig. \ref{fig:complexity_vs_overhead}, we pictorially depict the trade-off between the computational complexity and bootstrap overhead of the various decoders when $k = 500$ and when their decoding failure rate is upper bounded by $\delta=0.1$ along with the RSD parameter $c = 1$.\footnote{\textcolor{black}{Note that the RSD parameters are fixed to $\delta = 0.1$ and $c = 1$ for the presentation of results. Since the behavior of the bootstrap overhead and the average computational complexity of the decoders are already well established, similar experiments can be conducted for the other values of $\delta$ and $c$.}} In this depiction, the blue circle in the extreme right corner represents the BP decoder which clearly has the lowest computational complexity, however, a very high bootstrap overhead. It is vice-versa for the OFG decoder, which is represented using the red square. Labels in black and magenta denote the overhead-complexity pairs of the family of BRH decoders with their $K_{BR}$ values ranging from $K_{OFG}$ to $K_{BP}$, and the family of CRH decoders with their $\eta_{c}$ values ranging from 0 to $k$, respectively. It can be observed that the respective families of the BRH and CRH decoders act as bridges between the OFG and BP decoders, in terms of the overhead and complexity trade-off. Therefore, depending on whether the bucket node wants to operate at a fixed bootstrap overhead or fixed computational complexity, it can choose between the BRH and CRH decoders. If the bucket node chooses BRH decoder, it needs to select the parameter $K_{BR}$ that allows it to achieve its desired combination of computational complexity and bootstrap overhead. On the other hand, if CRH decoder is chosen, the bucket node needs to select the parameter $\eta_{c}$ that satisfies its complexity and overhead requirements.

In the next subsection, we define a new metric to equitably compare the performances of the various LT decoders discussed hitherto, based on which we optimize the performances of BRH and CRH decoders.

\begin{figure}
  \includegraphics[scale = 0.32]{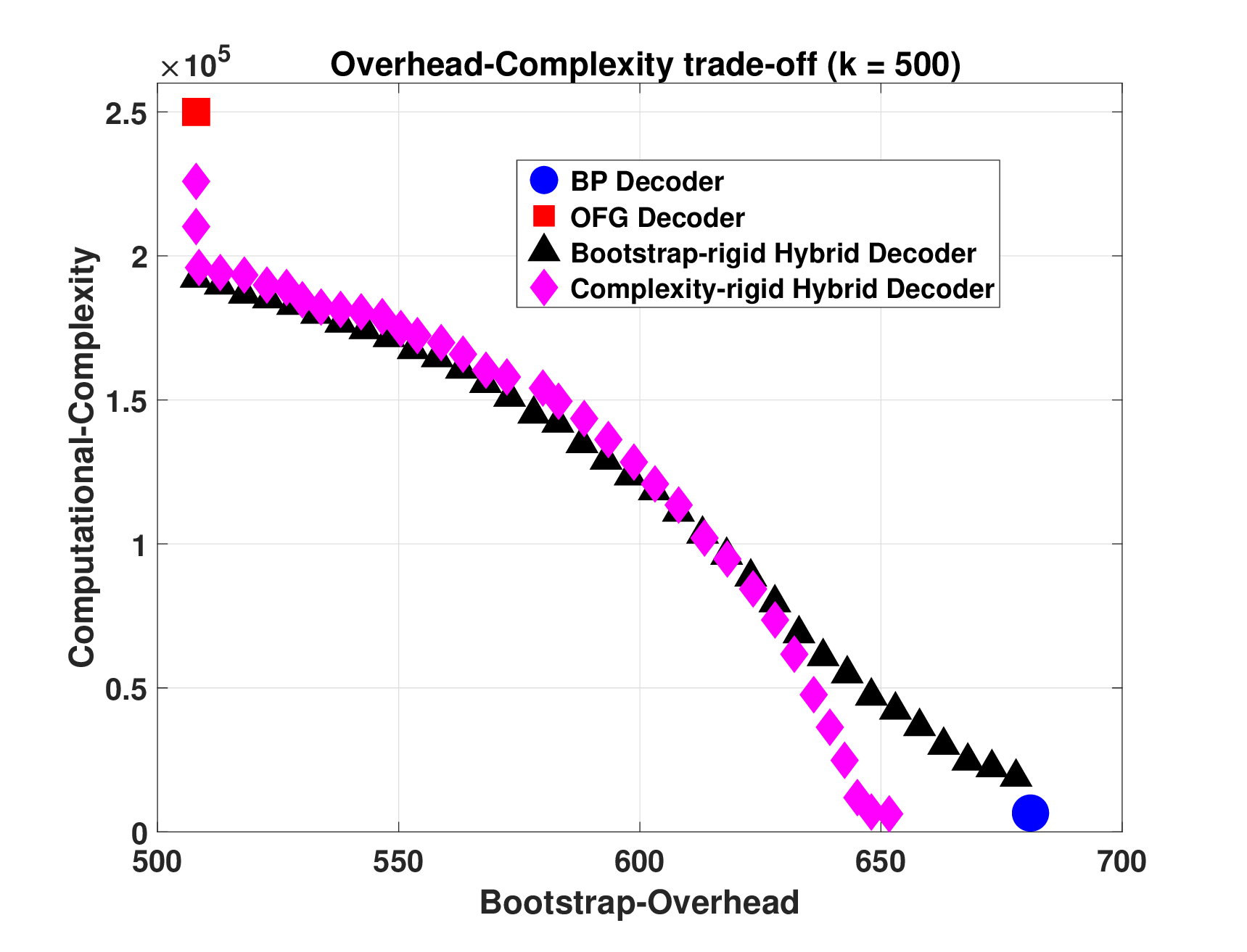}
  \vspace{-0.8cm}
  \caption{\textcolor{black}{BRH and CRH decoders acting as bridge between BP and OFG in terms of bootstrap overhead vs computational complexity trade-off.}}
  \label{fig:complexity_vs_overhead}
\end{figure}

\subsection{The Mirroring Cost}
A bucket node is always associated with a feasible range of computational complexity and bootstrap overhead over which it could operate. Out of the multiple complexity-overhead pairs which falls within this feasible range, a bucket node has to select one pair which gives the optimal performance. This could be done only after taking into consideration the costs associated with the bucket node in performing computations as well as downloading data. In this line, we define a new metric known as the \emph{absolute mirroring cost}, which encompasses the costs associated with computation and download for the bucket node, along with the computational complexity and communication overhead of the decoder. Henceforth, we use the absolute mirroring cost as the standard metric to compare the performances of the LT decoders. To quantify this metric, we represent \textcolor{black}{$c_1>0$} as the cost associated with the bucket node performing one XOR operation between two droplets. Similarly, \textcolor{black}{$c_2>0$} represents the cost for the bucket node to contact a droplet node and download one droplet from it. With that, the absolute mirroring cost is defined as follows:  


\begin{definition}[Absolute Mirroring Cost]\label{Mirroring Cost}
The absolute mirroring cost at a bucket node for decoding an epoch is defined as $c_{1}\times$\{computational complexity\} + $c_{2}\times$\{bootstrap overhead\}.
\end{definition}

In practice, $c_{1}$ and $c_{2}$ can be quantified as the time taken by the bucket node to perform one XOR operation between two droplets and the time taken by the bucket node to download one droplet, respectively. As a result, the absolute mirroring cost can also be quantified as the time taken by the bucket node to decode one epoch.  

In the rest of this section, we analyse the absolute mirroring cost using the cost ratio $\alpha$, where $\alpha = \frac{c_{2}}{c_{1}}$. In particular, we normalise $c_{1}$ to 1 and evaluate the normalized absolute mirroring costs of the decoders as functions of $\alpha$. We will henceforth refer to the normalized absolute mirroring cost as ``mirroring cost". From Definition \ref{Mirroring Cost}, \textcolor{black}{for a given $\alpha>0$}, the mirroring costs of BP and OFG decoders denoted by $M_{BP}(\alpha)$ and $M_{OFG}(\alpha)$ respectively, are:
\begin{equation} \label{eqn:Mirroring cost BP}
M_{BP}(\alpha) = K_{BP}\cdot D + \alpha\cdot K_{BP}.
\end{equation}
\vspace{-3mm}
\begin{equation} \label{eqn:Mirroring cost OFG}
\begin{aligned}
M_{OFG}(\alpha) = k^{2} + \alpha\cdot K_{OFG}.
\end{aligned}
\end{equation}

\subsubsection{Mirroring Costs of BRH and CRH Decoders}

\textcolor{black}{For a given $\alpha>0$}, the mirroring costs of BRH and CRH decoders indicated by $M_{BR}(\alpha)$ and $M_{CR}(\alpha)$, respectively, are given in 
\eqref{Norm_mirror_cost_BRH} and \eqref{Norm_mirror_cost_CRH}, respectively.

\begin{equation}\label{Norm_mirror_cost_BRH}
M_{BR}(\alpha) = K_{BR}\cdot D \cdot \frac{\EX[\eta_{B}]}{k}+\EX[(k-\eta_{B})^2]+\alpha\cdot K_{BR}.    
\end{equation}

\begin{equation}\label{Norm_mirror_cost_CRH}
M_{CR}(\alpha) = K_{CR}^{av}\cdot D \cdot \frac{\eta_{c}}{k}+(k-\eta_{c})^2+\alpha\cdot K_{CR}^{av}.    
\end{equation}

It is important to note that the mirroring costs of BRH and CRH decoders are not fixed for a given $\alpha$, unlike the BP and OFG decoders, which have static mirroring costs. This is due to the BRH decoder's variable $K_{BR}$, which can range from $K_{OFG}$ to $K_{BP}$, and CRH decoder's variable $\eta_{c}$ which can range from $0$ to $k$. 
Therefore, we propose an optimization problem (Problem \ref{problem:Optimization Problem}) for the BRH decoder to obtain the optimal number of droplet nodes that must be contacted by the bucket node, denoted by $K_{min}(\alpha)$, in order to achieve the minimum mirroring cost of the decoder. Similarly, we also propose an optimization problem (Problem \ref{problem:Optimization Problem CRH}) for the CRH decoder to obtain the optimal $\eta_{c}$, denoted by $\eta_{min}(\alpha)$, in order to achieve the minimum mirroring cost of the decoder.


\begin{mdframed}
\begin{problem}\label{problem:Optimization Problem}
For a given $\alpha > 0$, solve
\begin{equation} 
\begin{aligned}
K_{min}(\alpha) =& \,\underset{K_{BR}}{\arg\min}\left\{K_{BR} D \frac{\EX[\eta_{B}]}{k}+\EX[(k-\eta_{B})^2]+\alpha K_{BR}\right\} \\ \nonumber
\textrm{s.t.}\quad & K_{OFG}\le K_{BR}\le K_{BP}.\;
\end{aligned}
\end{equation}
\end{problem}
\end{mdframed}

\begin{mdframed}
\begin{problem}\label{problem:Optimization Problem CRH}
For a given $\alpha > 0$, solve
\begin{equation} 
\begin{aligned}
\eta_{min}(\alpha) =& \,\underset{\eta_{c}}{\arg\min}\left\{K_{CR}^{av}\cdot D \cdot \frac{\eta_{c}}{k}+(k-\eta_{c})^2+\alpha\cdot K_{CR}^{av}\right\} \\ \nonumber
\textrm{s.t.}\quad & 0\le \eta_{c}\le k.\;
\end{aligned}
\end{equation}
\end{problem}
\end{mdframed}

\begin{figure}
  \includegraphics[scale = 0.29]{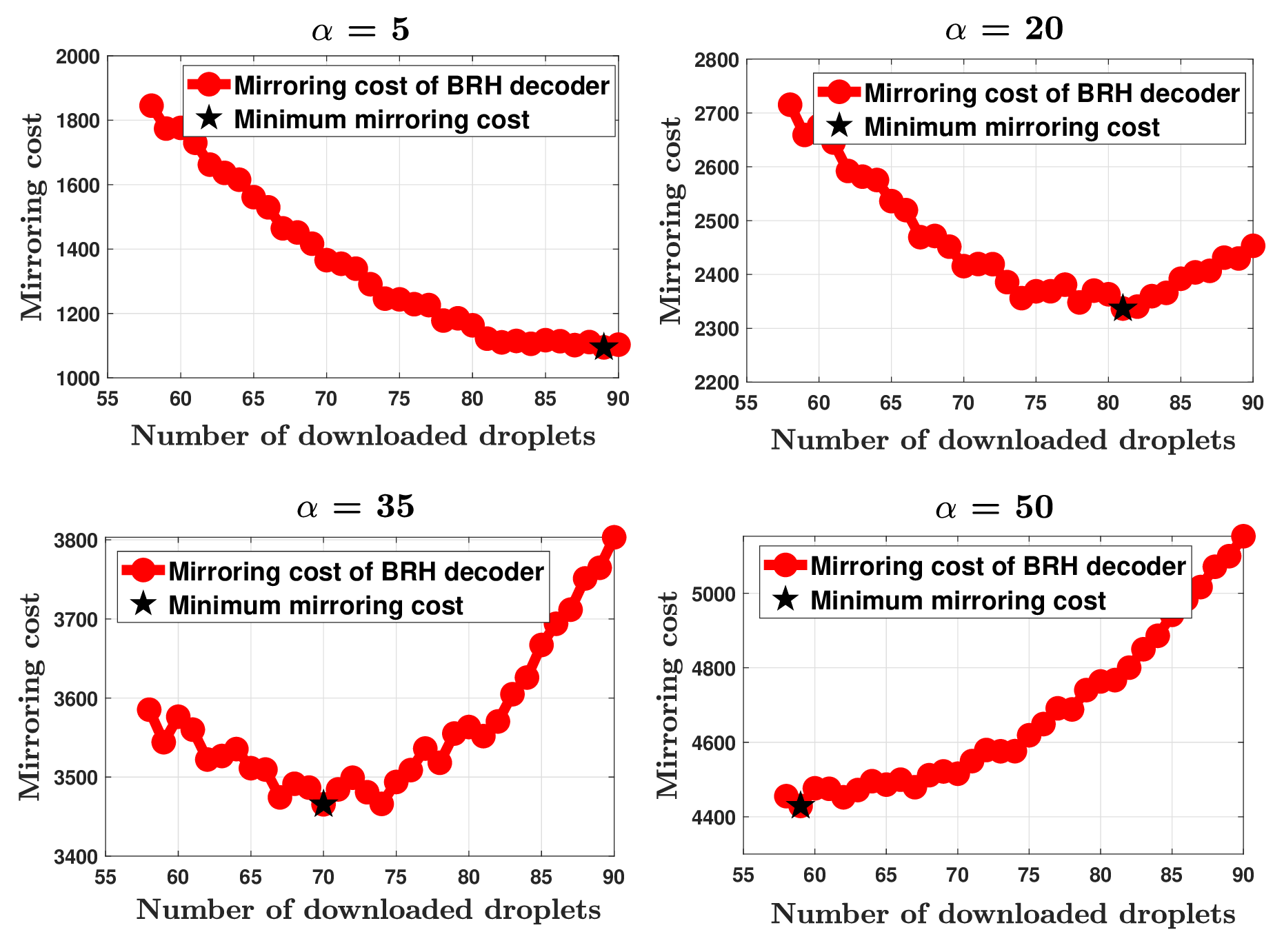}
  \vspace{-0.7cm}
  \caption{\textcolor{black}{Plots depicting the behaviour of the objective function of Problem \ref{problem:Optimization Problem} as a function of $\alpha$ for BRH decoder.}}
  \label{fig:Optimization}
\end{figure}

The behaviour of the objective functions captured in Problem \ref{problem:Optimization Problem} and Problem \ref{problem:Optimization Problem CRH} are depicted in Fig. \ref{fig:Optimization} and Fig. \ref{fig:Optimization_Complexity_Rigid}, respectively, for $k=50$.
The red circular points in Fig. \ref{fig:Optimization} and Fig. \ref{fig:Optimization_Complexity_Rigid} represent the evaluations of the objective functions in Problem \ref{problem:Optimization Problem} and Problem \ref{problem:Optimization Problem CRH}, respectively, for different values of $K_{BR}$ and $\eta_{c}$, respectively. The minimum mirroring costs of the BRH and CRH decoders, along with their respective $K_{min}(\alpha)$ and $\eta_{min}(\alpha)$, are also labelled in Fig. \ref{fig:Optimization} and Fig. \ref{fig:Optimization_Complexity_Rigid}, respectively. 
Furthermore, we can infer from both Fig. \ref{fig:Optimization} and Fig. \ref{fig:Optimization_Complexity_Rigid} that for a given $k$, as the value of $\alpha$ increases, the minimum point shifts towards the left. This is because an increase in $\alpha$ indicates a relative increase in $c_{2}$ with respect to $c_{1}$. Therefore, for the BRH decoder, $K_{BR}$ must be reduced to achieve the smallest mirroring cost, resulting in the leftward shift of $K_{min}(\alpha)$ as seen in Fig. \ref{fig:Optimization}.
Similarly, for the CRH decoder, $K_{CR}^{av}$ must be decreased because this contributes to the overhead term of the mirroring cost. To achieve this, $\eta_{c}$ must decrease, resulting in a decrease in $\eta_{min}(\alpha)$ as $\alpha$ increases.





\begin{figure}[ht!]
\begin{center}
  \includegraphics[scale = 0.30]{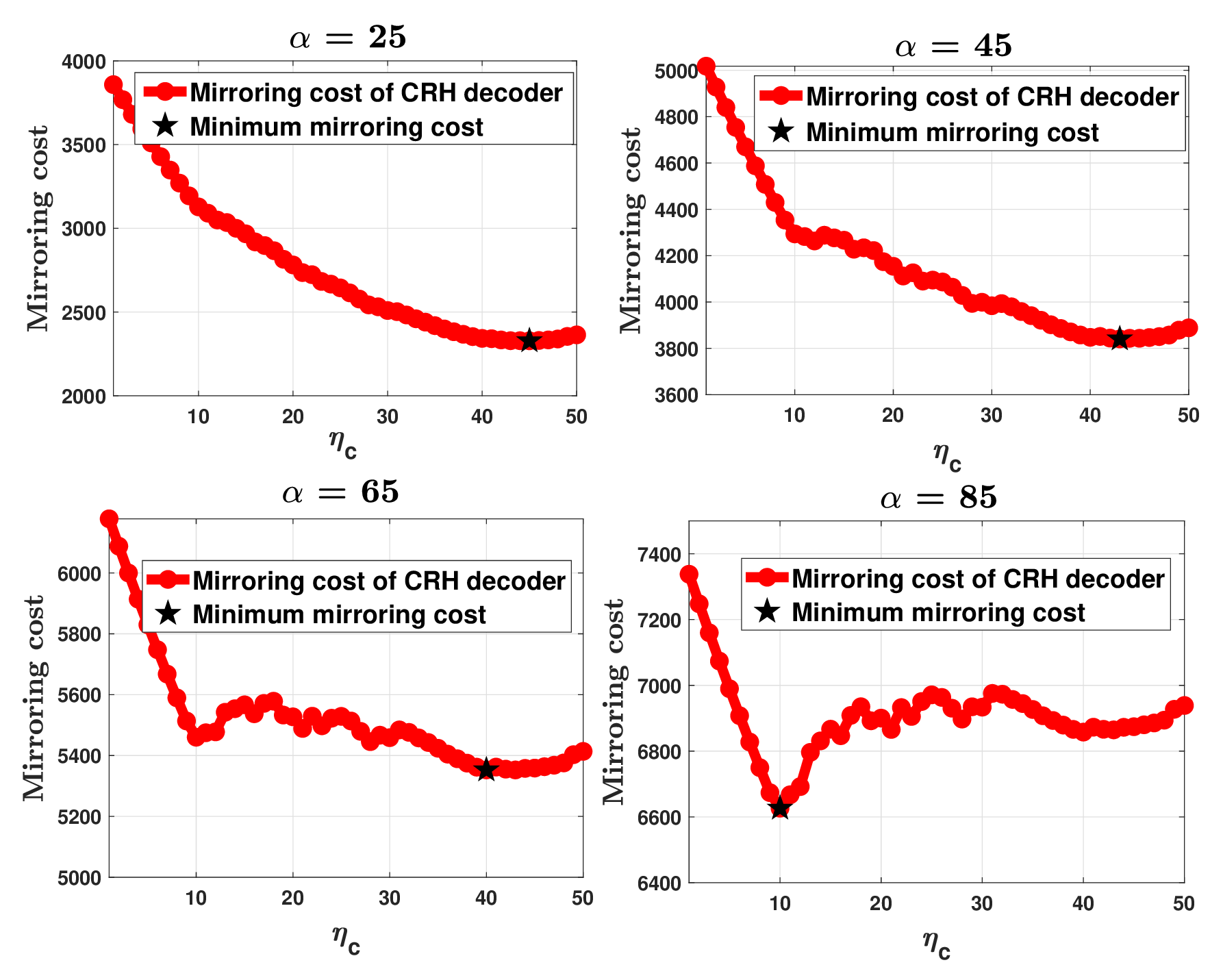}
  \end{center}
  \vspace{-0.7cm}
  \caption{\textcolor{black}{Plots depicting the behaviour of the objective function of Problem \ref{problem:Optimization Problem CRH} as a function of $\alpha$ for CRH decoder.}}
  \label{fig:Optimization_Complexity_Rigid}
\end{figure}

\begin{center}
\begin{table}[h] \caption{Numerical solutions to minimization problems of BRH and CRH decoders and their respective minimum mirroring costs (in 1000s)}\label{table:Optimization without Adversary}
\vspace{2mm}
{%
\addtolength{\tabcolsep}{-3.1pt}
\renewcommand{\arraystretch}{1.8}
\begin{center}
\begin{tabular}{|c|c|c|c|c|c|c|c|c|c|}
\hline
\multicolumn{5}{|c|}{$k=10$}
& 
\multicolumn{5}{|c|}{$k=50$} \\ \hline
\rule{0pt}{10pt}${\alpha}$ & $K_{min}$ & $\eta_{min}$ & $M_{BR}^{min}$ & $M_{CR}^{min}$ & ${\alpha}$ & $K_{min}$ & $\eta_{min}$ & $M_{BR}^{min}$ & $M_{CR}^{min}$  \\ \hline
$1$ & $15$ & $5$ & $0.081$ & $0.07$ & $1$  & $91$ & $45$ & $0.637$ & $0.51$   \\ 
$2$ & $15$ & $5$ & $0.096$ & $0.085$ & $15$ & $91$ & $45$ & $1.910$ & $1.57$  \\ 
$3$ & $15$ & $5$ & $0.111$ & $0.101$ & $30$ & $74$ & $43$ & $3.096$ & $2.71$ \\ 
$4$ & $15$ & $5$ & $0.126$ & $0.116$ & $45$ & $59$ & $43$ & $4.134$ & $3.84$ \\ 
$5$ & $15$ & $5$ & $0.141$ & $0.131$ & $60$ & $59$ & $43$ & $5.019$ & $4.97$ \\ \hline
\multicolumn{5}{|c|}{$k=100$}
& 
\multicolumn{5}{|c|}{$k=500$} \\ \hline
\rule{0pt}{10pt}
${\alpha}$ & $K_{min}$ & $\eta_{min}$ & $M_{BR}^{min}$ & $M_{CR}^{min}$ & ${\alpha}$ & $K_{min}$ & $\eta_{min}$ & $M_{BR}^{min}$ & $M_{CR}^{min}$  \\ \hline
$1$  & $163$ & $95$ & $1.313$ & $1.134$ & $1$ & $681$ & $492$ & $7.387$ & $7.01$  \\ 
$35$ & $163$ & $95$ & $6.854$ & $6.03$ & $600$ & $681$ & $478$ & $415.3$ & $395.7$ \\ 
$70$ & $163$ & $95$ & $12.56$ & $11.07$ & $1200$ & $507$ & $455$ & $800.6$ & $784.5$ \\ 
$105$ & $136$ & $90$ & $17.78$ & $16.1$ & $1800$ & $507$ & $57$ & $1104$ & $1109$ \\ 
$140$ & $109$ & $90$ & $21.91$ & $21.14$ & $2400$ & $507$ & $57$ & $1409$ & $1414$ \\ \hline
\end{tabular}
\end{center}
}
\end{table}
\end{center}

\begin{figure}[ht!]
  \includegraphics[scale = 0.30]{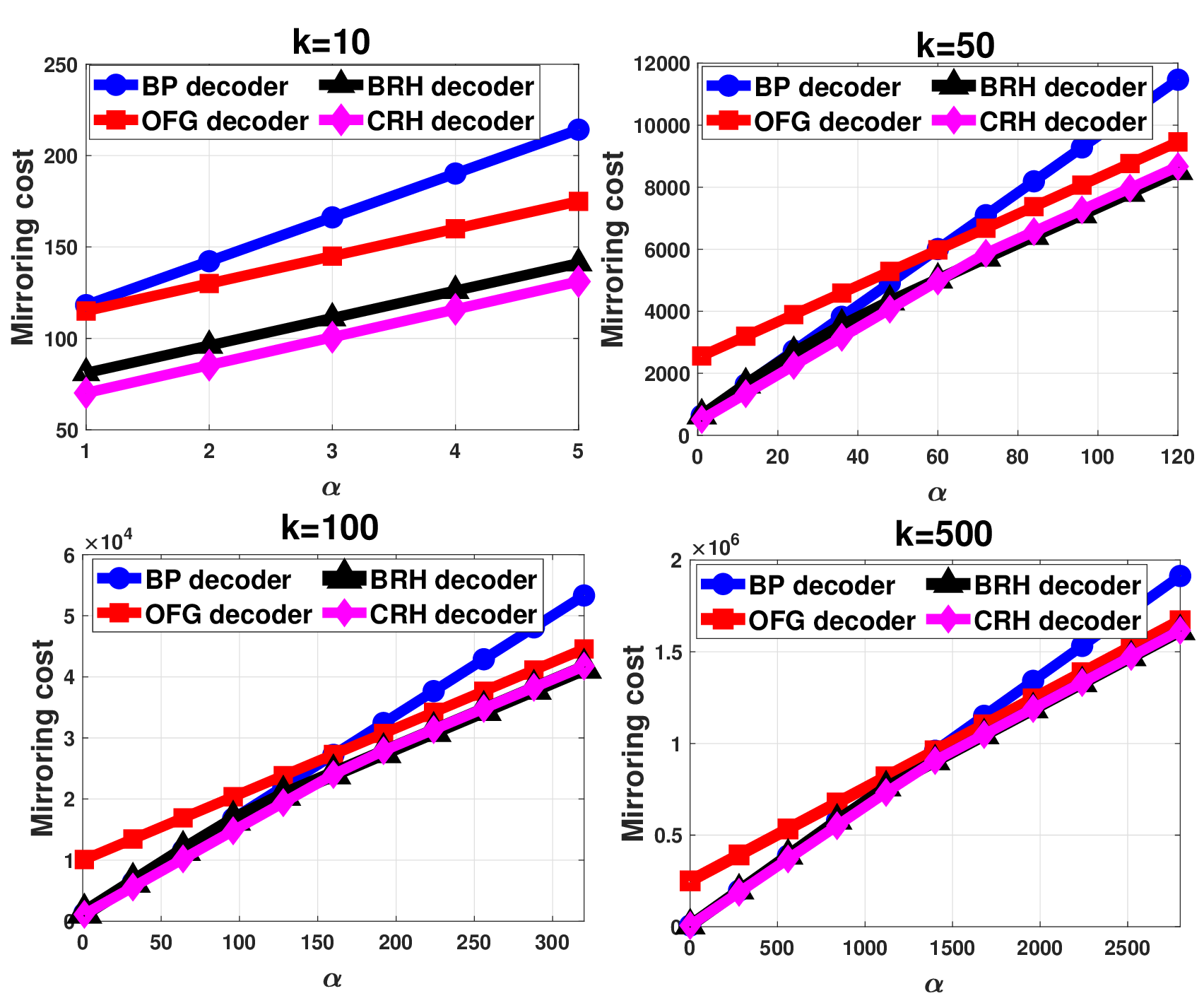}
  \vspace{-0.3cm}
  \caption{\textcolor{black}{Mirroring cost comparison of BP, OFG, BRH and CRH decoders.}}
  \label{fig:Overall Decoding cost}
\end{figure}

\subsubsection{Simulation Results}

In this section, we present the numerical solutions to Problem \ref{problem:Optimization Problem} and Problem \ref{problem:Optimization Problem CRH} for different values of $\alpha$ in Table \ref{table:Optimization without Adversary}, for values of $k \in \{10,50,100,500\}$. We also evaluate the minimum mirroring costs of the BRH and CRH decoders, denoted by $M_{BR}^{min}$ and $M_{CR}^{min}$ respectively, in Table \ref{table:Optimization without Adversary}.
Furthermore, we show the comparison of the mirroring costs of the four decoders as functions of $\alpha$ in Fig. \ref{fig:Overall Decoding cost}, for values of $k \in \{10,50,100,500\}$. The parameters of LT code used for simulation are $c = 0.1$ and $\delta = 0.1$. Fig. \ref{fig:Overall Decoding cost} shows that for a given $\alpha$, the minimum mirroring costs of both BRH and CRH decoders are always less than or equal to those of the BP and OFG decoders, which is achieved by optimizing the BRH and CRH decoders.

\subsection{Discussion}
In this section, we have proposed two variants of hybrid decoders for LT codes, namely the BRH and the CRH decoders, both of which offer flexibility to a bucket node by allowing it to operate on its feasible range of computational complexity and bootstrap overhead. We have further optimized the proposed BRH and CRH decoders, taking into account the bucket node's computational and download expenses. Overall, we emphasize that a bucket node can choose to employ either the BRH or the CRH decoder, depending on whether it requires a constant bootstrap overhead or a constant computational complexity. Then the bucket node may appropriately select its $K_{min}(\alpha)$ (for BRH) or $\eta_{min}(\alpha)$ (for CRH) depending on its cost ratio $\alpha$, to decode the blockchain at the lowest mirroring cost.
\section{Denial-of-Service Threats on LT Coded Blockchains}
\label{sec:threats}

\begin{figure*}
\begin{center}
   \includegraphics[scale = 0.4]{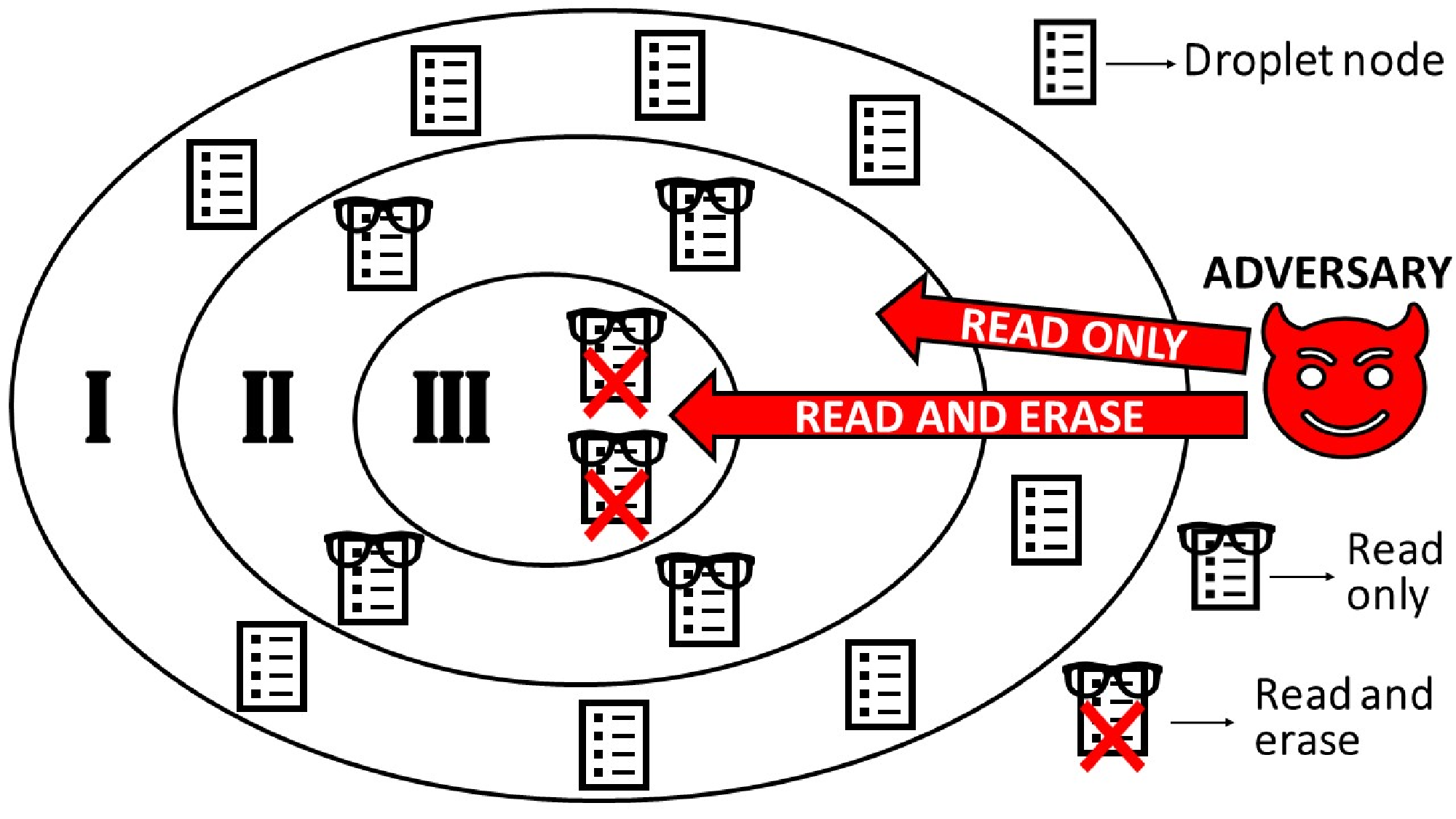}
  \end{center}
  \vspace{-0.5cm}
  \caption{The proposed non-oblivious adversarial model is executed in two stages. In the first stage, a set of droplet nodes of a full node is read by the adversary wherein the information on the degrees of the droplet nodes are retrieved (marked by region II). In the second stage, a subset of the read nodes is erased based on the information on their degrees as well as the decoder employed by the victim bucket nodes (marked by region III).}
  \label{fig:Adversary_model}
\end{figure*}

\textcolor{black}{In the preceding sections, we have shown that LT coded blockchains support low-complexity algorithms for the bucket nodes to download and mirror the transactions of the blockchains. As a consequence, bucket nodes, despite possessing limited resources in terms of computational workload and bandwidth, can join the blockchain thereby facilitating seamless scalability. Despite these benefits of LT coded blockchains, they may also be vulnerable to a wide range of security threats, akin to any other networked systems. In the context of blockchains, some relevant security threats could be those degrading the infrastructure scalability or compromising the ability to perform legitimate transactions.}

\textcolor{black}{In this work, we are interested in studying the vulnerabilities of LT coded blockchains against DoS (denial of service) threats that pose a critical risk to the scalability of LT coded blockchains, disrupting their ability to operate effectively in decentralized environments. In particular, our vulnerability analysis is on scenarios where the attacker's objective is to design strategies that render a subset of the droplet nodes unavailable for other bucket nodes to download the blockchain. As a consequence, new nodes will not be able to join the blockchain. Thus, the number of nodes that store the complete history of transactions in the blockchain will not grow thereby limiting the number of honest nodes to achieve consensus for verifying the  transactions. Although the direct consequence of the attacker's objective is to limit the scalability of LT coded blockchain, its indirect and long-term impact is to curtail the network's ability for consensus due to the presence of fewer number of full nodes.}

\textcolor{black}{In the following section, we present the assumptions we make with respect to the capabilities of the attacker, in terms of resources to execute the attacks, and the implementation methods through which the attacks could be executed in practice.}

\subsection{\textcolor{black}{Threat Model}}

\textcolor{black}{We consider an adversary or a group of adversaries that launch DoS attacks \citep{DoS_1,DoS_2,Ransomware_healthcare,Ransomware_CryptoLocker,Ransomware_Petya} on a subset of droplet nodes, thereby preventing new bucket nodes from using the compromised droplet nodes for their decoding process. As a consequence, blockchain scalability is hindered, as new bucket nodes who want to join the network may not receive sufficient data to retrieve the original blockchain. To launch DoS threats, an adversary may  flood a targeted set of droplet nodes with repeated service requests, such that a legitimate bucket node is not able to access the compromised droplet nodes \citep{DoS_1}. As an alternative, an adversary could also make the droplet nodes unavailable by compromising them akin to the methods used in ransomware attacks. We remark that the way an adversary makes the droplet nodes unavailable is out of the scope of this work. Instead, we argue that the way in which an adversary chooses the target droplet nodes to compromise is more relevant. This is because if the choice of the target droplet nodes is made based on the contents stored in them, then it might have a higher impact.}

\textcolor{black}{Towards choosing the target droplet nodes, one can classify the threat model into two categories, namely: (i) \emph{Oblivious model} and (ii) \emph{Non-oblivious model}. An adversary employing the oblivious model, referred to as an oblivious adversary, does not observe the storage contents of any of the droplet nodes before performing the attack. Consequently, the oblivious adversary blindly launches DoS attack on a random subset of droplet nodes. Henceforth, we refer to an attack executed by an oblivious adversary as \emph{blind attack}. On the other hand, non-oblivious threat model refers to a scenario wherein an adversary gains access to a set of droplet nodes, reads their contents, and uses this information to select a subset of the accessed nodes to attack. In practice, a typical way for the adversary to gain access into the droplet nodes is to enter the network as a potential bucket node and then download their contents. After reading the contents of the droplet nodes, it could then decide to attack a subset of them.}

\textcolor{black}{The authors of \citep{SeF} have shown that LT coded blockchains are resilient against such blind attacks. However, to the best of our knowledge, the impact of a non-oblivious adversarial model on LT-coded blockchains has yet to be investigated. Pointing at this research gap, in this work, we adopt a non-oblivious model, wherein the adversary gains access to a set of droplet nodes, reads their contents, and uses this information to select a subset of the accessed nodes to attack. To realize a non-oblivious model, we assume that the adversary in our threat model has sufficient resources to gain access into a set of droplet nodes, and read their contents. In practice, when the adversary joins the network as a potential bucket node, the resources required for this task corresponds to the communication-overheads for I/O access and download. Subsequently, based on the contents in the droplets, we assume that the adversary has necessary additional resources required to make the droplets unavailable, either through ransomware- or flooding-based attacks. Given that the resources for accessing the contents and executing attacks are attacker dependent, we parameterize such capabilities, and then study their impact on LT coded blockchains. Henceforth, we refer to the attacked droplet nodes as \emph{erased} nodes, and the residual droplet nodes that are not attacked as \emph{honest} nodes. In the subsequent section, we show that the selection of the subset of droplet nodes to be attacked is determined by the decoder employed by the bucket node.}

\subsection{DoS Threats from Non-Oblivious Adversaries}


Let us consider a full node that contains $S$ number of droplet nodes under the LT coded architecture. Recall from Section \ref{sec2_model} that the degree information of droplet nodes is available in the form of a binary vector $\mathbf{v}$, which is stored along with the coded blocks. On this full node, we consider a non-oblivious adversary who acts as a legitimate bucket node and reads the degree information of a randomly chosen $\sigma_{0}$-fraction of $S$ droplet nodes, where $\sigma_{0} \in (0, 1)$. 
He then uses this information to tactically choose a subset of $\sigma S$ droplet nodes to erase among them, where $0<\sigma\le \sigma_{0}$. Fig. \ref{fig:Adversary_model} depicts our adversarial model wherein regions \Romannum{1}, \Romannum{2} and \Romannum{3} respectively represent the set of droplet nodes that the adversary has (i) neither read nor erased, (ii) read, however not erased, and (iii) both read and erased. As a special case, note that our adversarial model will collapse to the oblivious model if $\sigma_{0} = \sigma$. This is equivalent to the adversary attacking a randomly chosen $\sigma$-fraction of $S$ droplet nodes without reading the degree information on their coded structures. 


In the rest of this section, we explore questions on what could be a reasonable \emph{attack strategy} from the adversary's perspective, given that he knows what decoder is used by the legitimate bucket nodes for blockchain retrieval. In this context, given the set of $\sigma_{0} S$ droplet nodes read by the adversary, a reasonable attack strategy refers to an appropriate selection of a subset of $\sigma S$ droplet nodes to erase, such that the probability of decoding failure at the legitimate bucket node resulting from this strategy is significantly higher than that of a blind attack.  
To this end, we propose novel reasonable attack strategies for each of the BP, OFG, BRH and CRH decoders.

\subsubsection{Attack Strategy against BP Decoder}

Before presenting a reasonable attack strategy against the BP decoder, we recall that the BP decoder thrives on the availability of the singleton droplets and then iteratively generates new singletons to recover the blockchain. Note that the adversary neither has control on the specific droplet nodes that would be accessed by the victim bucket node, nor knows the order in which the singletons would be used by the BP decoder. Therefore, under such conditions, a reasonable attack strategy for the adversary would be to reduce the likelihood of the victim bucket node accessing singletons. With that, we present a reasonable attack strategy against the BP decoder in the following.

\begin{strategy}[Degree based attack strategy]
\thlabel{BP_Optimal_Attack}
A reasonable attack strategy against the BP decoder is to arrange the $\sigma_{0} S$ droplet nodes in the increasing order of their degrees and then erase the first $\sigma S$ droplet nodes of the sorted set. Note that the order of the droplet nodes having the same degree can be chosen in an arbitrary manner.
\end{strategy}

In the above strategy, the objective of the adversary is to exhaust the singletons before the BP decoder retrieves all blocks of the epoch. It is intuitive that the probability that a degree-$d$ droplet node attains the singleton state sooner in the BP decoding process, decreases with increase in $d$. Therefore, among the droplets nodes read by the adversary, the singletons and the droplet nodes that attain the singleton state early in the BP decoding process must be erased first. As a result, a reasonable attack strategy is to erase the droplet nodes in the increasing order of their degrees. 

Although the aforementioned degree based attack strategy effectively exploits the weakness of the BP decoder, a drawback of this strategy is that it considers droplet nodes with the same degree to be equally valuable to the BP decoder. However, it does not take into account the information about the neighbours of the droplet nodes. Recall that neighbours of a droplet node refer to the collection of input blocks of an epoch that were XORed to obtain that droplet node. In particular, if an adversary needs to erase a subset of droplet nodes, all having the same degree $d$ for some $d\in [k]$, \thref{BP_Optimal_Attack} instructs him to erase a subset of them chosen at random. However, the droplet nodes that are not erased by the adversary may be more valuable to the BP decoder than those that are erased, thereby benefiting the legitimate bucket nodes.

To address this issue, we propose a score based attack strategy, where individual scores are assigned to each droplet node read by the adversary, based on their degrees and neighbours. The scores are assigned such that the droplet node most valuable to the BP decoder obtains the highest score, with the score decreasing as the droplet node's priority declines.

\begin{strategy}[Score based attack strategy]
\thlabel{BP_Score_Attack}
The score based attack strategy against the BP decoder involves the following steps.
\begin{itemize}
    \item The adversary begins by assigning initial scores to all the $\sigma_{0}S$ droplet nodes read by him. In particular, for a droplet node with degree $d \in \{1, 2, \ldots, k\}$, its initial score is set to $\left(\frac{1}{2}\right)^{(d-1)}$.
    \item Assuming that the $\sigma_{0}S$ droplet nodes are ordered in some fashion, the adversary sets the iteration number $i = 1$ and updates the scores of each droplet node as follows.
    \begin{enumerate}
        \item The adversary picks the $i^{th}$ droplet node, denoted by $\mathbf{x}_i$. Let the degree of this droplet node be $d_{i}$, and its current score be $\mathcal{S}(\mathbf{x}_i)$. Then, he identifies the droplet nodes whose degrees are greater than $d_{i}$. Let that corresponding set of droplet nodes be denoted by the set $\mathcal{X}_{d_{i}}$. 
        \item For each droplet node $\mathbf{y} \in \mathcal{X}_{d_{i}}$ with degree $\Tilde{d}$, if the neighbours of $\mathbf{x}_i$ are a subset of the neighbours of $\mathbf{y}$, the adversary updates the score of $\mathbf{x}_i$ as $\mathcal{S}(\mathbf{x}_i)$ $\gets$ $\mathcal{S}(\mathbf{x}_i)$ + $\left(\frac{1}{2}\right)^{(\Tilde{d}-1)}$. The adversary repeats this step for all $\mathbf{y} \in \mathcal{X}_{d_{i}}$.
        \item If $i \le \sigma_{0}S$, the adversary increments $i$ by 1 and goes to Step (1). Otherwise, the scoring process is terminated.
    \end{enumerate}
    \item Finally, the adversary sorts the droplet nodes in the decreasing order of their respective scores and then erases the first $\sigma S$ droplet nodes of this sorted set. Note that the order of the droplet nodes having the same score can be chosen in an arbitrary manner.
\end{itemize}
\end{strategy}

We further demonstrate the score based attack strategy using a simplified toy example in the following.

\begin{example}
Let us consider a full node of LT coded blockchain with a total of $S=20$ droplet nodes and an epoch size of $k = 6$. In particular, we consider a non-oblivious adversary for which $\sigma_{0} = 0.5$ and $\sigma = 0.25$. Therefore, the adversary can read $10$ droplet nodes and erase $5$ of them. Furthermore, in this setup, we represent the droplet nodes with their corresponding length-$k$ binary vectors as they contain information about the droplet nodes' degrees and neighbours. The binary vectors of the droplet nodes read by the adversary are presented as columns in Table \ref{table:Binary_Vectors_Read}. 

\vspace*{-\baselineskip}
\begin{center}
\begin{table}[h] \caption{Binary vectors of droplet nodes read by the adversary, along with their respective initial scores} \label{table:Binary_Vectors_Read}
{%
\addtolength{\tabcolsep}{-1.5pt}
\renewcommand{\arraystretch}{1.75}
\begin{center}
\begin{tabular}{|c|c|c|c|c|c|c|c|c|c|c|} 
\hline
\textbf{Droplet nodes} & $\mathbf{x}_1$ & $\mathbf{x}_2$ & $\mathbf{x}_3$ & $\mathbf{x}_4$ & $\mathbf{x}_5$ & $\mathbf{x}_6$ & $\mathbf{x}_7$ & $\mathbf{x}_8$ & $\mathbf{x}_9$ & $\mathbf{x}_{10}$ \\ \hline
Block 1 & $0$ & $0$ & $1$ & $0$ & $0$ & $0$ & $0$ & $1$ & $1$ & $1$\\
Block 2 & $0$ & $0$ & $0$ & $0$ & $0$ & $1$ & $1$ & $0$ & $1$ & $1$\\
Block 3 & $1$ & $0$ & $1$ & $0$ & $0$ & $0$ & $1$ & $1$ & $1$ & $1$\\
Block 4 & $0$ & $0$ & $0$ & $1$ & $0$ & $0$ & $0$ & $1$ & $0$ & $1$\\
Block 5 & $0$ & $1$ & $0$ & $0$ & $1$ & $1$ & $0$ & $0$ & $1$ & $1$\\
Block 6 & $0$ & $0$ & $0$ & $1$ & $1$ & $0$ & $1$ & $1$ & $0$ & $1$\\
\hline
\textbf{Initial scores} & $\mathbf{1}$ & $\mathbf{1}$ & $\mathbf{\frac{1}{2}}$ & $\mathbf{\frac{1}{2}}$ & $\mathbf{\frac{1}{2}}$ & $\mathbf{\frac{1}{2}}$ & $\mathbf{\frac{1}{4}}$ & $\mathbf{\frac{1}{8}}$ & $\mathbf{\frac{1}{8}}$ & $\mathbf{\frac{1}{32}}$\\ \hline

\end{tabular}
\end{center}
}
\end{table}
\end{center}

In the first step, the adversary calculates the initial scores of droplet nodes $x_i$ for $i\in[1,10]$, as $\left(\frac{1}{2}\right)^{(d-1)}$,  where $d$ is the degree of the $i^{th}$ droplet node. These initial scores are presented in Table \ref{table:Binary_Vectors_Read}.
In the next step, the adversary starts with $\mathbf{x}_1$ to perform score updation. 
Since the degree of $\mathbf{x}_1$ is $1$, $\mathcal{X}_{1} = \{ \mathbf{x}_i ~|~ i \in [3,10]\}$. We observe that the neighbours of $\mathbf{x}_1$ are a subset of the neighbours of $\mathbf{x}_3$, $\mathbf{x}_7$, $\mathbf{x}_8$, $\mathbf{x}_9$ and $\mathbf{x}_{10}$ of the set $\mathcal{X}_{1}$, whose degrees are $2$, $3$, $4$, $4$ and $6$ respectively. As a result, the updated score of $\mathbf{x}_1$ is $1+(\frac{1}{2}+\frac{1}{4}+\frac{1}{8}+\frac{1}{8}+\frac{1}{32})$. Similarly, the updated scores of the remaining droplet nodes, i.e. $\mathbf{x}_i$ for $i \in [2,10]$ are also calculated. The final scores of droplet nodes read by the adversary, $\mathbf{x}_i$ for $i \in [1,10]$ are \{$2.03125$, $2.15625$, $0.78125$, $0.65625$, $0.53125$, $0.65625$, $0.28125$, $0.15625$, $0.15625$, $0.03125$\}, respectively. Finally, the adversary arranges the droplet nodes in the decreasing order of their respective scores to obtain the new order as \{$\mathbf{x}_2$, $\mathbf{x}_1$, $\mathbf{x}_3$, $\mathbf{x}_4$, $\mathbf{x}_6$, $\mathbf{x}_5$, $\mathbf{x}_7$, $\mathbf{x}_8$, $\mathbf{x}_9$, $\mathbf{x}_{10}$\}. He then erases the first $5$ droplet nodes of this ordered set.
\end{example}

In Section \ref{Simulation_Attack_Strategies}, we will show that the score based attack strategy performs better than the degree based attack strategy in terms of enhancing the failure rate at a legitimate bucket node.

\subsubsection{Attack Strategy against OFG Decoder}
Along the lines of the BP decoder, even with the OFG decoder, the adversary has neither control over the specific droplet nodes accessed by the bucket node nor knowledge of the order in which gaussian elimination is used. Therefore, under such conditions, a reasonable attack strategy for the adversary is to reduce the probability with which the bucket node can access droplet nodes with rank $k$ on the $\mathbf{G}$ matrix.

\begin{strategy}\thlabel{OFG_Optimal_Attack}
A reasonable attack strategy against the OFG decoder involves the following steps:
\begin{enumerate}
\item Using the $\sigma_{0}S$ droplet nodes, the adversary juxtaposes the binary row vectors $\mathbf{v}$ corresponding to each droplet node and forms a binary matrix of dimensions $\sigma_{0}S \times k$.  
\item Then, he identifies $(\sigma_{0}-\sigma)S$ rows of the above matrix that will result in the minimum rank. In this context, we refer to the matrix's rank as the rank calculated over the binary field.
\item Finally, he erases the residual $\sigma S$ droplet nodes which correspond to the complementary rows of the above step.
\end{enumerate}
\end{strategy}

The idea behind the strategy is that for an OFG decoder to fail, the matrix derived from the binary row vectors of the honest droplet nodes must become rank deficient. To accomplish this, a reasonable approach is to leave those droplet nodes read by the adversary unerased, whose coefficients produce a matrix of least possible rank.

One possible method to carry out the aforementioned attack strategy is to undertake an exhaustive search to find the combination of droplet nodes that produces a matrix of minimum rank. However, this is a computationally expensive approach as it involves the evaluation of ranks of $\sigma_{0}S \choose \sigma S$ matrices, each with dimensions $(\sigma_{0}-\sigma)S \times k$.   
To circumvent this problem, we present a computationally feasible algorithm, presented in Algorithm \ref{alg:Minimum rank}, as an alternative approach to exhaustive search. We emphasize that the algorithm may not identify the exact $(\sigma_{0}-\sigma)S$ droplet nodes of minimum rank, however, it does provide a solution with low rank.

\begin{algorithm}
\caption{Computationally efficient algorithm for carrying out the attack strategy against OFG decoder}\label{alg:Minimum rank}
\begin{algorithmic}[1]
    \Require{Matrix $\mathbf{R} \in \{0,1\}^{\sigma_{0}S \times k}$ whose rows correspond to binary vectors of the droplet nodes read by the adversary} 
    \Ensure{Matrix $\mathbf{M} \in \{0,1\}^{(\sigma_{0}-\sigma)S \times k}$ whose rows correspond to the set of droplet nodes with a reasonably low rank}

    \State Initialize $\mathbf{M} \gets \mathbf{0}$;
    \State $\mathbf{M}[1][\cdot]$ $\gets$ $\mathbf{R}[1][\cdot]$;
		\For {$i \gets 2$ to $(\sigma_{0}-\sigma)S$}
            \State r $\gets$ Rank($\mathbf{M}$);
			\For {$j \gets 2$ to $\sigma_{0}S-i+2$}
				\State $\mathbf{M}[i][\cdot] \gets \mathbf{R}[j][\cdot]$;
				\If {Rank($\mathbf{M}$) == r}
                    \State Swap($\mathbf{R}[j][\cdot]$,$\mathbf{R}[\sigma_{0}S-i+2][\cdot]$);
                    \State \textbf{break};
                \EndIf
			\EndFor
		\EndFor
    \end{algorithmic} 
\end{algorithm}

The input of Algorithm \ref{alg:Minimum rank} is the matrix $\mathbf{R} \in \{0,1\}^{\sigma_{0}S \times k}$, whose rows correspond to binary vectors of the droplet nodes read by the adversary. The algorithm outputs a matrix $\mathbf{M} \in \{0,1\}^{(\sigma_{0}-\sigma)S \times k}$, whose rows correspond to a set of droplet nodes with a reasonably low rank. With that, the sequence of steps involved in the algorithm are enumerated as follows.
\begin{enumerate}
    \item First, we initialize a matrix $\mathbf{M}$ of dimension $(\sigma_{0}-\sigma)S \times k$, and set all its entries to zero.
    
    \item In the next step, we replace the first row of $\mathbf{M}$ by the first row of $\mathbf{R}$. The rest of the rows of $\mathbf{R}$ are referred to as unused rows. We also set the iteration number, $i=2$.
    
    \item When choosing the entries for the $i^{th}$ row of $\mathbf{M}$, we look for an unused row in $\mathbf{R}$ that does not enhance the rank of $\mathbf{M}$, when that row replaces the $i$-th row of $\mathbf{M}$.
    \begin{enumerate}
        \item If such a row exists in $\mathbf{R}$, we replace that row as the $i^{th}$ row of $\mathbf{M}$, and refer this row of $\mathbf{R}$ as a used row.
        \item Otherwise, we replace the $i^{th}$ row of $\mathbf{M}$ by an arbitrary unused row of $\mathbf{R}$, and refer this row of $\mathbf{R}$ as a used row.
    \end{enumerate}
    \item Increment $i$ by $1$.
    \item If $\mathbf{M}$ has all-zero rows, go to Step (3). 
\end{enumerate}

Finally, the adversary obtains a matrix $\mathbf{M}$ with a reasonably low rank by carrying out the aforementioned steps. He then erases the residual droplet nodes, whose binary coefficients correspond to the set of rows in $\mathbf{R}$ that are not present in $\mathbf{M}$, leaving behind a set of droplet nodes with a reasonably low rank for the legitimate bucket nodes.

In Algorithm \ref{alg:Minimum rank}, the maximum number of rank evaluations in the $i^{th}$ iteration is $\sigma_{0}S-i+1$. Likewise, there are $(\sigma_{0}-\sigma)S$ number of iterations in total. As a result, the algorithm's total number of rank evaluations is limited to $\sum_{i=1}^{(\sigma_{0}-\sigma)S} (\sigma_{0}S-i+1)$, which can be further approximated to $\sigma_{0}S \cdot (\sigma_{0}-\sigma)S$. Therefore, it can be concluded that the implementation of Algorithm \ref{alg:Minimum rank} requires substantially lower computational effort than that of an exhaustive search. 

\subsubsection{Attack Strategy against BRH Decoder}
When using the BRH decoder, recall that the bucket node is flexible with its computational complexity. As a result, owing to the possibility of implementing the OFG decoder, minimizing the rank of the $\mathbf{G}$ matrix corresponding to the honest droplet nodes is a reasonable attack strategy.  

\begin{strategy}\thlabel{BFG_Optimal_Attack}
A reasonable attack strategy against the BRH decoder is same as that against the OFG decoder.
\end{strategy}

\subsubsection{Attack Strategy against CRH Decoder}

With the CRH decoder, it is clear that its requirements for successfully decoding of the blockchain are in terms of both singleton droplet nodes as well as rank offered by the residual honest droplet nodes. As a result, putting forward a reasonable attack strategy for this case is non-trivial. However, given that the BP part of the CRH decoder is executed before the OFG part, we believe that depleting singletons in the network would help stalling the decoding process. Thus, with the following strategy, we show that attacks in \thref{BP_Optimal_Attack} and \thref{OFG_Optimal_Attack} are applicable depending on the value of $\eta_{c}$. Recall that $\eta_{c}$ is the number of blocks of an epoch that must be retrieved by the BP part of the CRH decoder.

\begin{strategy}\thlabel{CRH_Optimal_Attack}
For the CRH decoder with $\eta_{c} > 0$ , we propose to follow the same attack strategy as that against the BP decoder. However, for the CRH decoder with $\eta_{c} = 0$ , we propose to follow the same attack strategy as that against the OFG decoder.
\end{strategy} 

\subsection{Experimental Results for Attack Strategies} 

\textcolor{black}{In this section, first, we provide details on our experimental setup, and then present the impact of the proposed attack strategies.} 

\subsubsection{\textcolor{black}{Experimental Setup}}

\textcolor{black}{All our experiments are conducted via simulations on a standalone machine by modelling the following three steps. \textbf{Step 1}: Generation of droplets as per the LT encoding process. \textbf{Step 2}: Execution of the proposed threat models on a subset of LT coded droplets based on the read and write capabilties of the attacker, and \textbf{Step 3}: Reconstruction of the blockchain by the  bucket nodes in the network using the coded fragments in the \emph{live} droplet nodes. In order to generate the experimental results related to decoding failure rates, we use Monte-Carlo simulations \citep{Monte_carlo} by implementing \textbf{Step 1}, \textbf{Step 2} and \textbf{Step 3} on MATLAB over an ensemble of $10^{4}$ iterations. In each iteration, to implement \textbf{Step 1}, we fix the LT code parameters for the blockchain, and then generate a certain number of droplets in a random manner using the robust Soliton distribution \citep{SeF}. For \textbf{Step 2}, we fix the read and write capabilities of the attacker, and randomly choose a certain number of droplets with uniform distribution. Subsequently, depending on the attack strategy chosen by the attacker, we make a subset of the read droplet nodes unavailable for other bucket nodes. We realize this task in our simulation model by erasing the contents of the droplet nodes. Finally, for \textbf{Step 3}, we use the rest of the \emph{live} droplet nodes to reconstruct the blockchain. Given that the droplet generation process is stochastic, the choice of the droplets read by the attacker is stochastic in each interaction, we declare a decoding failure if a bucket node is unable to reconstruct the entire $k$ blocks of blockchain in each iteration. Overall, by aggregating the events of decoding failures over $10^{4}$ iterations, we empirically compute the decoding failure rate, and plot them as a function of various system parameters in the experimental results. This experimental setup is used to generate all the results in the rest of the paper.} 

\subsubsection{\textcolor{black}{Results}}
\label{Simulation_Attack_Strategies}
In this section, we present the simulation results that demonstrate the effectiveness of the proposed attack strategies corresponding to each of the decoders discussed hitherto. To present the simulation results, let us define the read-write ratio for an adversary as $\xi \triangleq \frac{\sigma_{0}}{\sigma}$. As we have $\sigma_{0} \ge \sigma$, this implies $\xi\ge 1$.

\begin{figure}[ht!]
  \includegraphics[scale = 0.31]{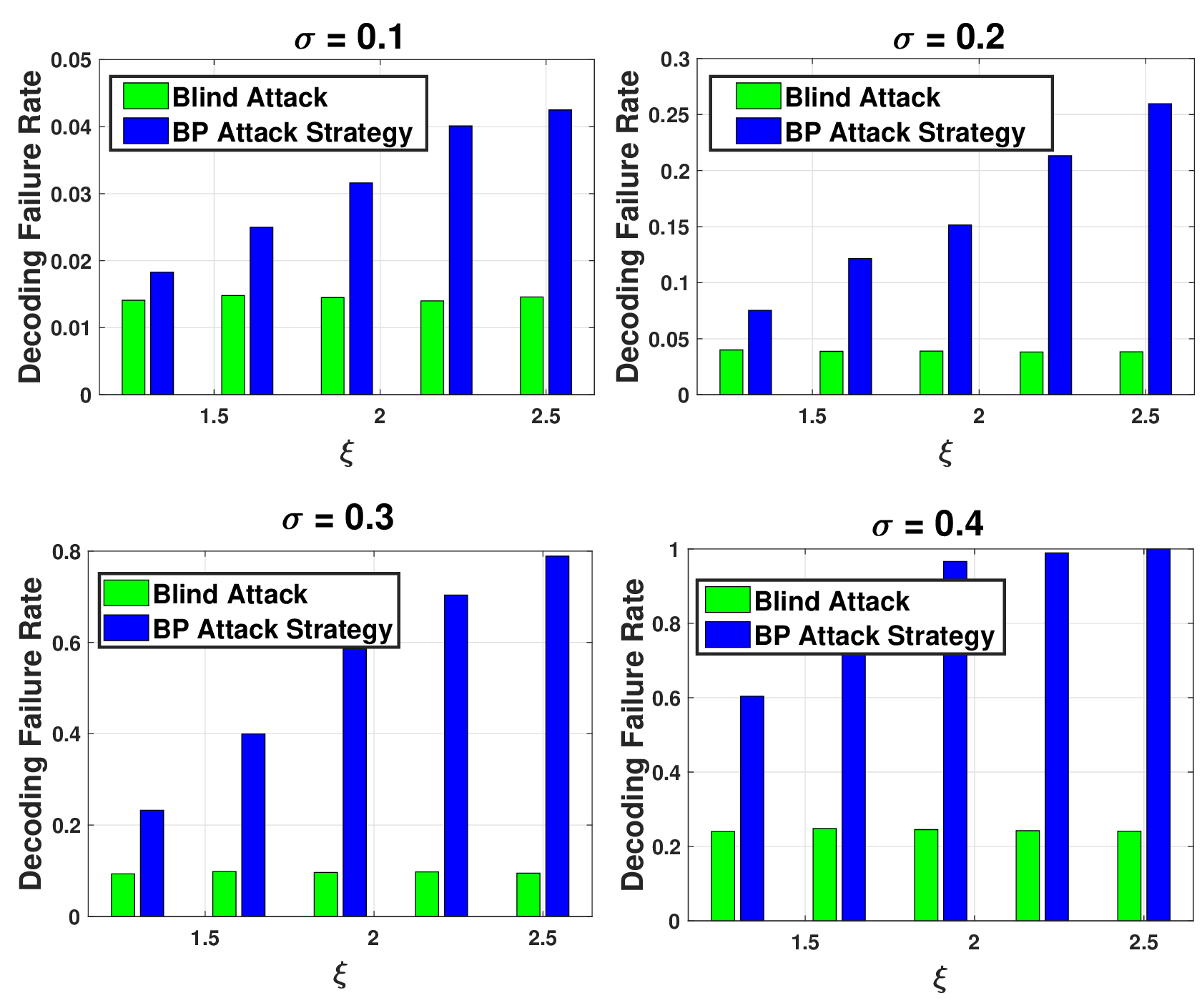}
  \vspace{-0.7cm}
  \caption{\textcolor{black}{Effectiveness of the proposed degree based attack strategy against a bucket node with BP decoder. The parameters for simulations are $k$=20, $S$=60, $\sigma \in \{0.1,0.2,0.3,0.4\}$ and $\xi \in \{1.3,1.6,1.9,2.2,2.5\}$.}}
  \label{fig:Bar_graph_BP_attack}
\end{figure}

First, we conduct simulations to characterize the decoding failure rate of the BP decoder for a given $\sigma$, as a function of $\xi$. Our experiments are performed for $k = 20$, $S = 3k$ such that $\sigma \in \{0.1,0.2,0.3,0.4\}$. For each value of $\sigma$, decoding failure rates are presented with $\xi \in \{1.3,1.6,1.9,2.2,2.5\}$ for two cases: one wherein the droplet nodes that are erased are randomly chosen, and the other wherein droplet nodes are chosen after executing the attack strategy presented in \thref{BP_Optimal_Attack}. For each ($\sigma$,$\xi$) pair, the obtained failure rates are presented in Fig. \ref{fig:Bar_graph_BP_attack}. Based on the blue bars in Fig. \ref{fig:Bar_graph_BP_attack}, it is clear that for a given $\sigma$, decoding failure rate increases with increase in $\xi$ when the droplet nodes are erased by executing the degree based attack strategy. This behaviour is intuitive since the adversary should get advantage as he reads more droplet nodes. \textcolor{black}{Furthermore, when the erased droplet nodes are randomly chosen, referred to as blind attacks, we observe from the green bars that the failure rate remains constant at different values of $\xi$ for a given $\sigma$.} 

\begin{figure}[ht!]
  \includegraphics[scale = 0.31]{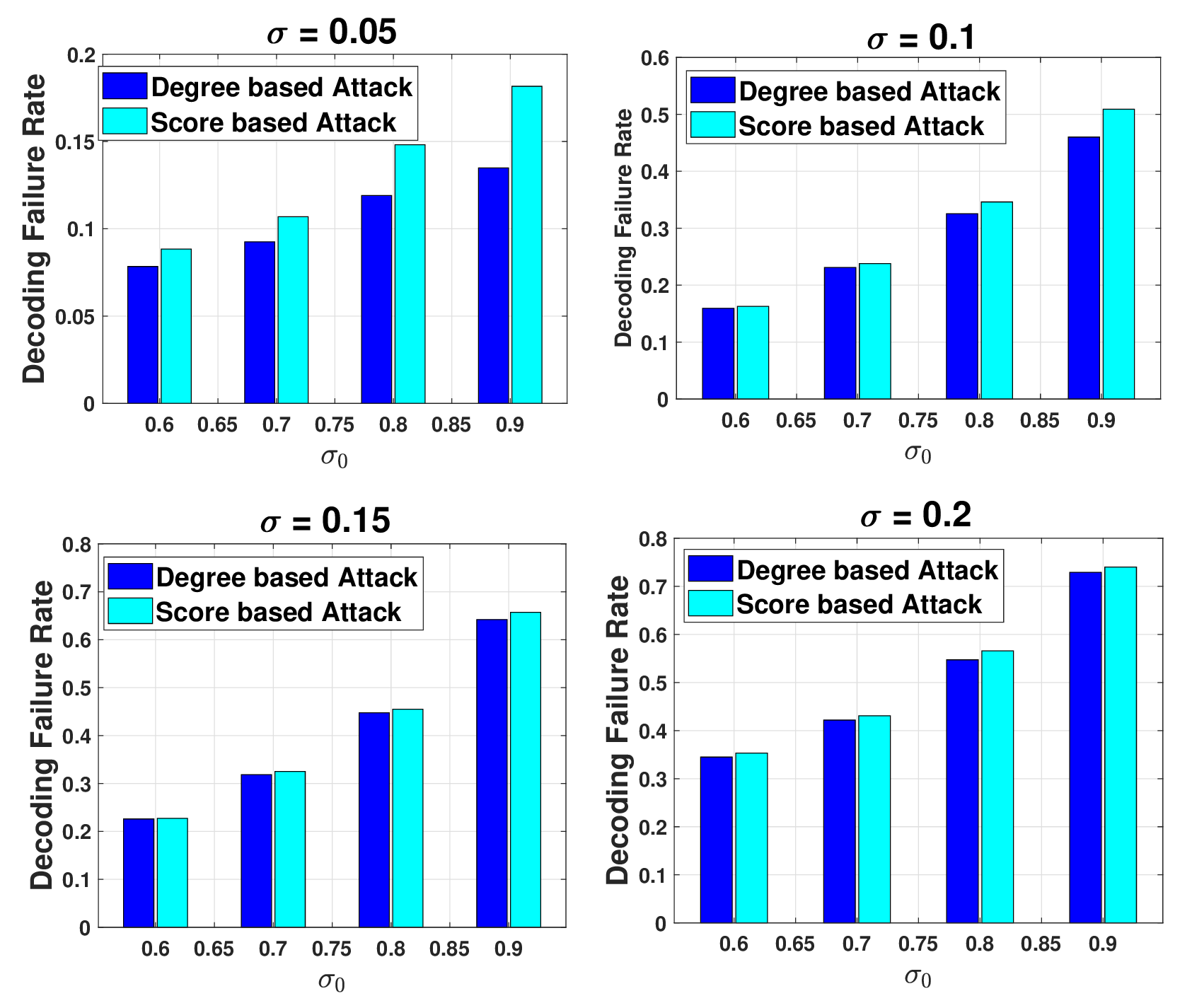}
  \vspace{-1cm}
  \caption{\textcolor{black}{Comparison of degree based and score based attack strategies against a bucket node with BP decoder. The parameters for simulations are $k$=20, $S$=60, $\sigma \in \{0.05,0.1,0.15,0.2\}$ and $\sigma_{0} \in \{0.6,0.7,0.8,0.9\}$.}}
  \label{fig:Bar_graph_deg_score_comparison}
\end{figure}

\begin{figure}[ht!]
  \includegraphics[scale = 0.30]{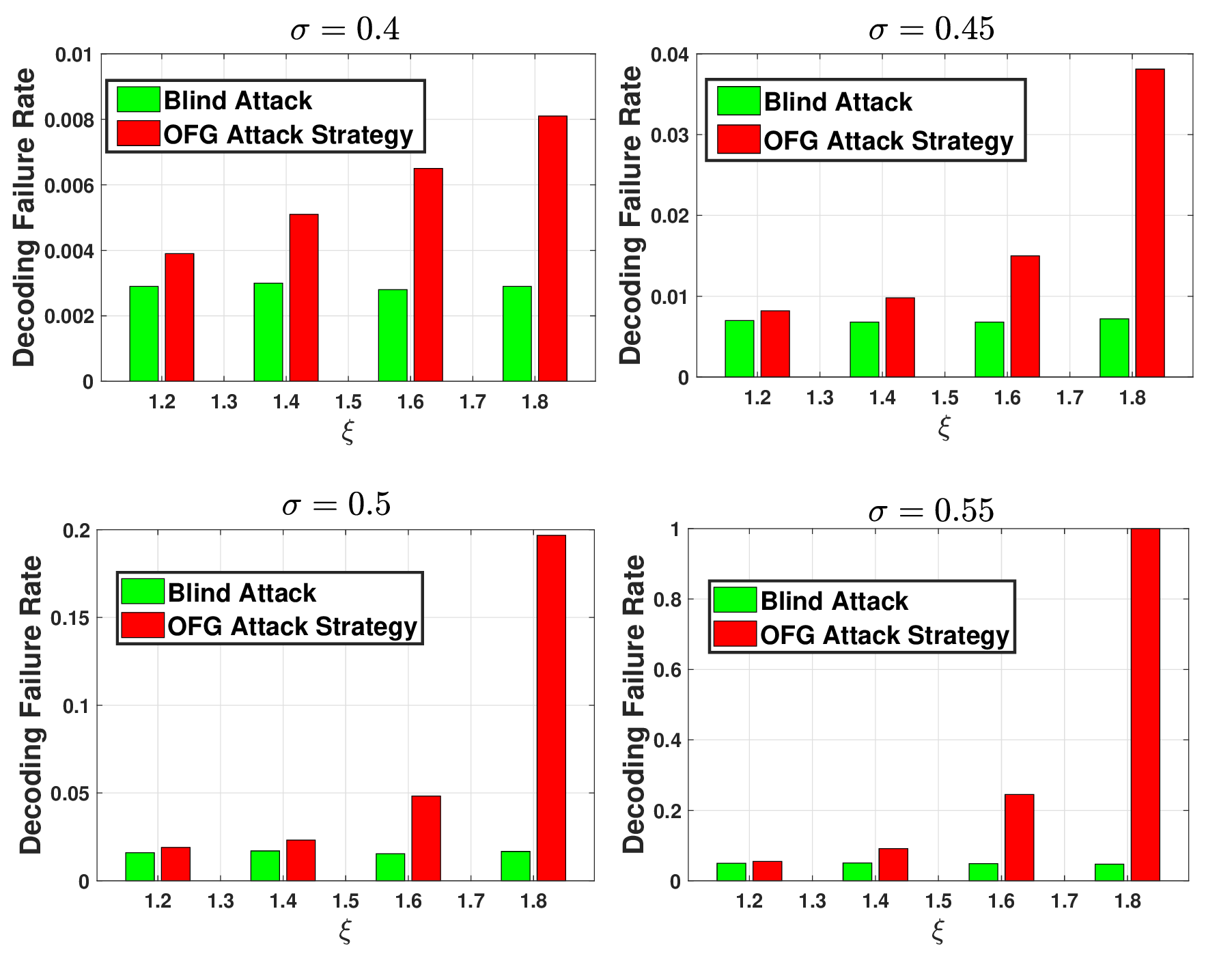}
  \vspace{-1cm}
  \caption{\textcolor{black}{Effectiveness of the proposed attack strategy against a bucket node with OFG decoder. The parameters for the simulations are $k$=20, $S$=60, $\sigma \in \{0.4,0.45,0.5,0.55\}$ and $\xi \in \{1.2,1.4,1.6,1.8\}$.}}
  \label{fig:Bar_graph_OFG_attack}
\end{figure}

Next, we present simulation results demonstrating that the score-based attack strategy for BP decoder, proposed in \thref{BP_Score_Attack} is more effective than the one described in \thref{BP_Optimal_Attack}, in Fig. \ref{fig:Bar_graph_deg_score_comparison}. The parameters used for simulation are specified in the caption of Fig. \ref{fig:Bar_graph_deg_score_comparison}. By comparing the blue and cyan bars in Fig. \ref{fig:Bar_graph_deg_score_comparison}, it is clear that there is an incremental improvement in the failure rate of score-based attack strategy compared to the degree based strategy. We can also observe that this increment is more significant for lower values of $\sigma$. This is intuitive because, as $\sigma$ approaches $\sigma_{0}$, the set of droplet nodes erased by an adversary starts becoming similar for both the strategies. As a result, their respective failure rates also turn out to be almost identical. However, at small $\sigma$, the score-based attack strategy erases the droplet nodes that are of highest priority to the BP decoder, whereas the degree based strategy may not erase the exact same droplet nodes. This creates a significant difference in the failure rates between the two strategies.


\begin{figure}[ht!]
  \includegraphics[scale = 0.30]{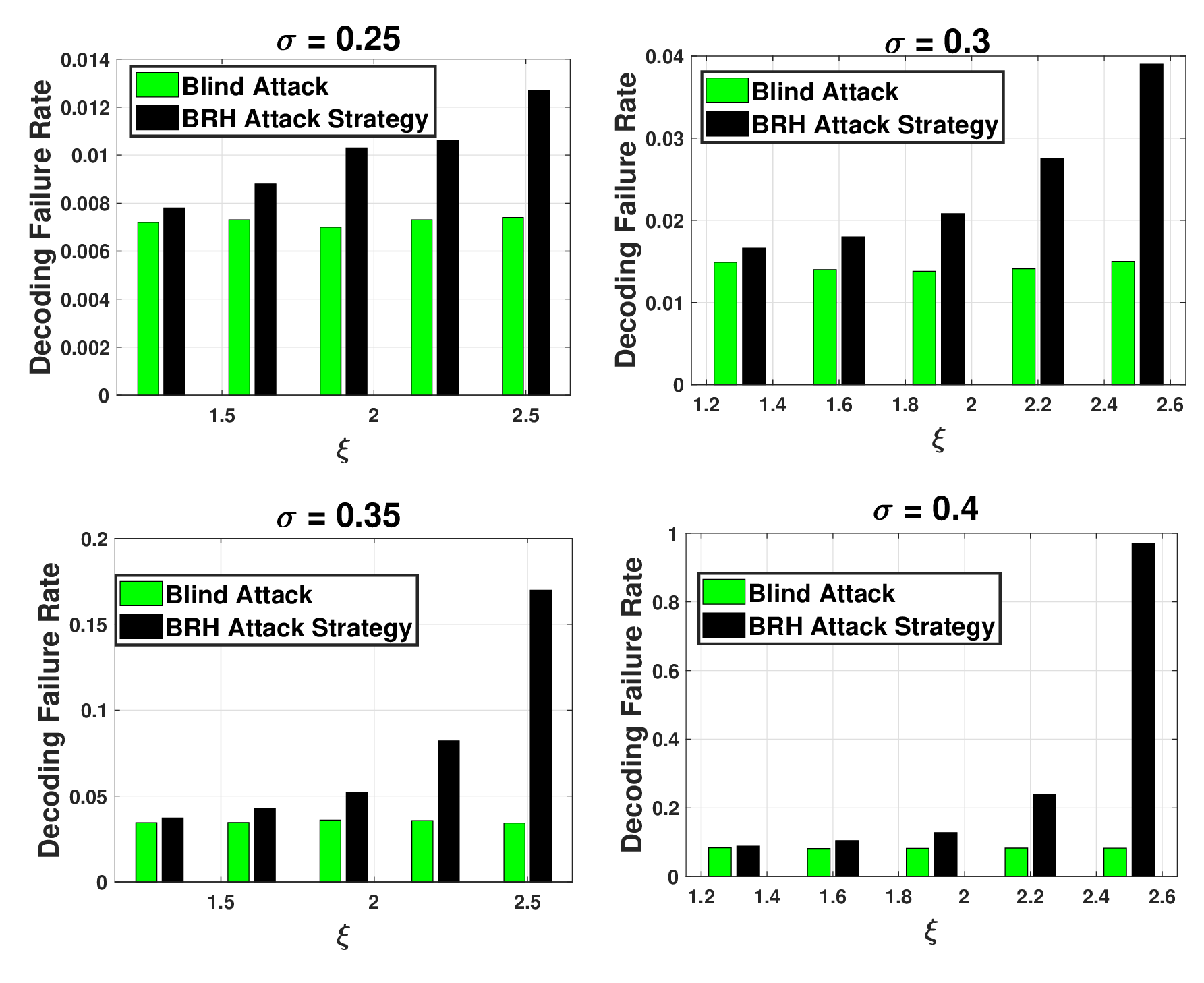}
  \vspace{-1cm}
  \caption{\textcolor{black}{Effectiveness of the proposed attack strategy against a bucket node with BRH decoder. The simulation parameters for the experiment are $k = 20$, $S = 3k$, $K = 44$, $\sigma \in \{0.4,0.45,0.5,0.55\}$ and $\xi \in \{1.3,1.6,1.9,2.2,2.5\}$.}}
  \label{fig:Bar_graph_BRH_attack}
\end{figure}

\begin{figure}[ht!]
  \includegraphics[scale = 0.29]{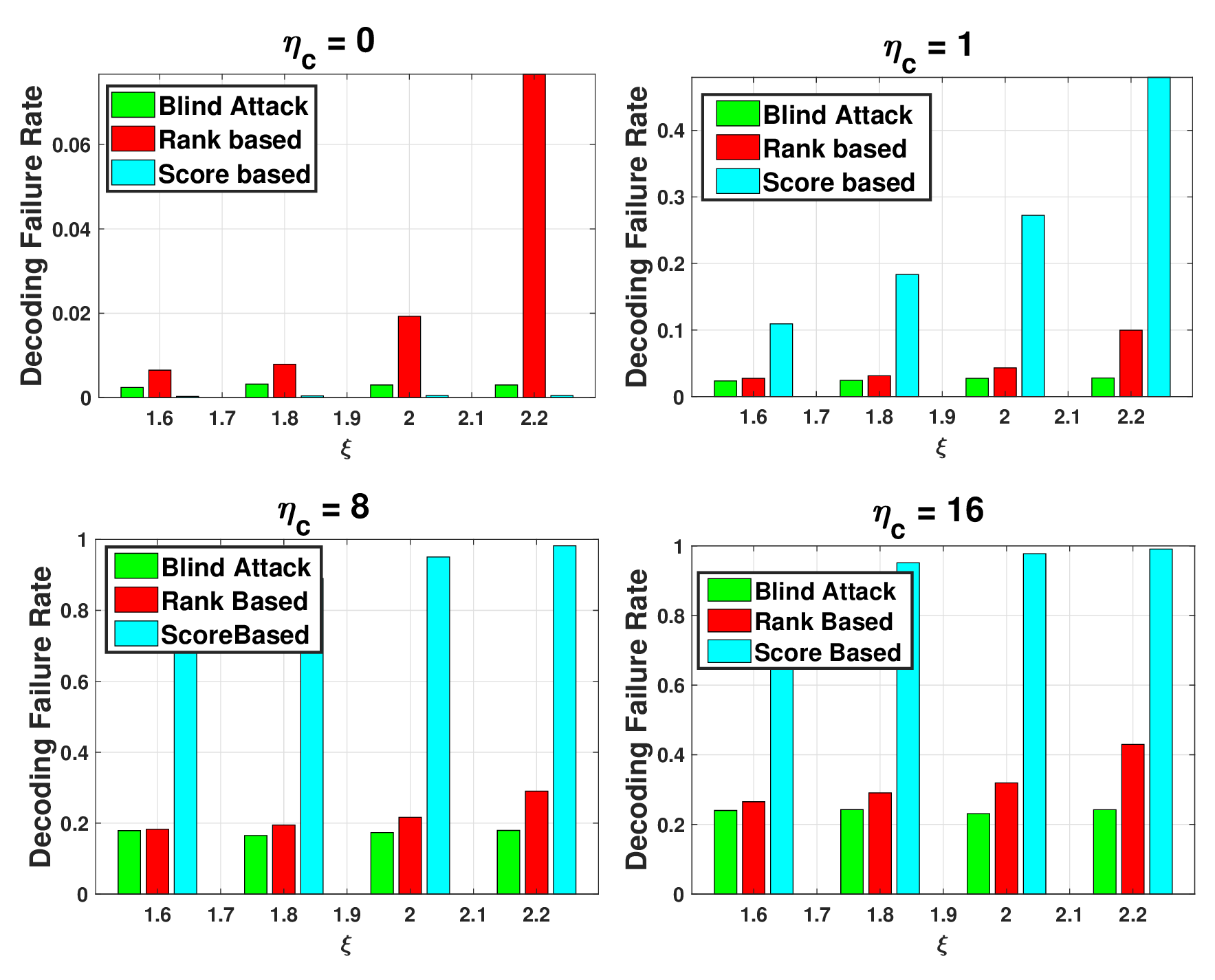}
  \vspace{-0.4cm}
  \caption{\textcolor{black}{Effect of the proposed attack strategies against the CRH decoder. The simulation parameters are $k$=20 and $S$=60. The legends for the top two subplots follow for the bottom two subplots.}}
  \label{fig:Bar_graph_CRH_attack}
\end{figure}

Similar experiments are conducted for the OFG and BRH decoders, and their results are presented in Fig. \ref{fig:Bar_graph_OFG_attack} and Fig. \ref{fig:Bar_graph_BRH_attack}, respectively. The inferences in this case are similar to that of the BP decoder. Note that in the above experiment, we implemented Algorithm \ref{alg:Minimum rank} to carry out attack strategies of both OFG and BRH decoders, as an exhaustive search is too expensive.

Finally, Fig. \ref{fig:Bar_graph_CRH_attack} demonstrates the comparison on the effectiveness of the attack strategies given in \thref{BP_Score_Attack} and \thref{OFG_Optimal_Attack} on the CRH decoder. The experiments are performed for $k=20$, $S=3k$ and $\sigma=0.3$. The failure rates of the CRH decoder corresponding to random erasing (green bars), minimizing rank strategy (red bars) and executing score-based attack strategy (cyan bars) are obtained for values of $\xi \in \{1.5,2,2.5,3\}$. The experiments were conducted for values of $\eta_{c} \in \{1,16\}$, in order to study the behaviour of the employed attack strategies, for near-extreme values of $\eta_{c}$. From the plots, we infer that for any non-zero values of $\eta_{c}$, the attack strategy in \thref{BP_Score_Attack} produces higher decoding failure rates. However, for $\eta_{c} = 0$, the attack strategy in \thref{OFG_Optimal_Attack} maximises the failure rate because the CRH decoder does not necessarily require any singleton for decoding and hence minimizing the rank of honest droplet nodes is sufficient.



\subsection{Cost Constrained Optimal Attacks}

\begin{figure}
  \includegraphics[scale = 0.36]{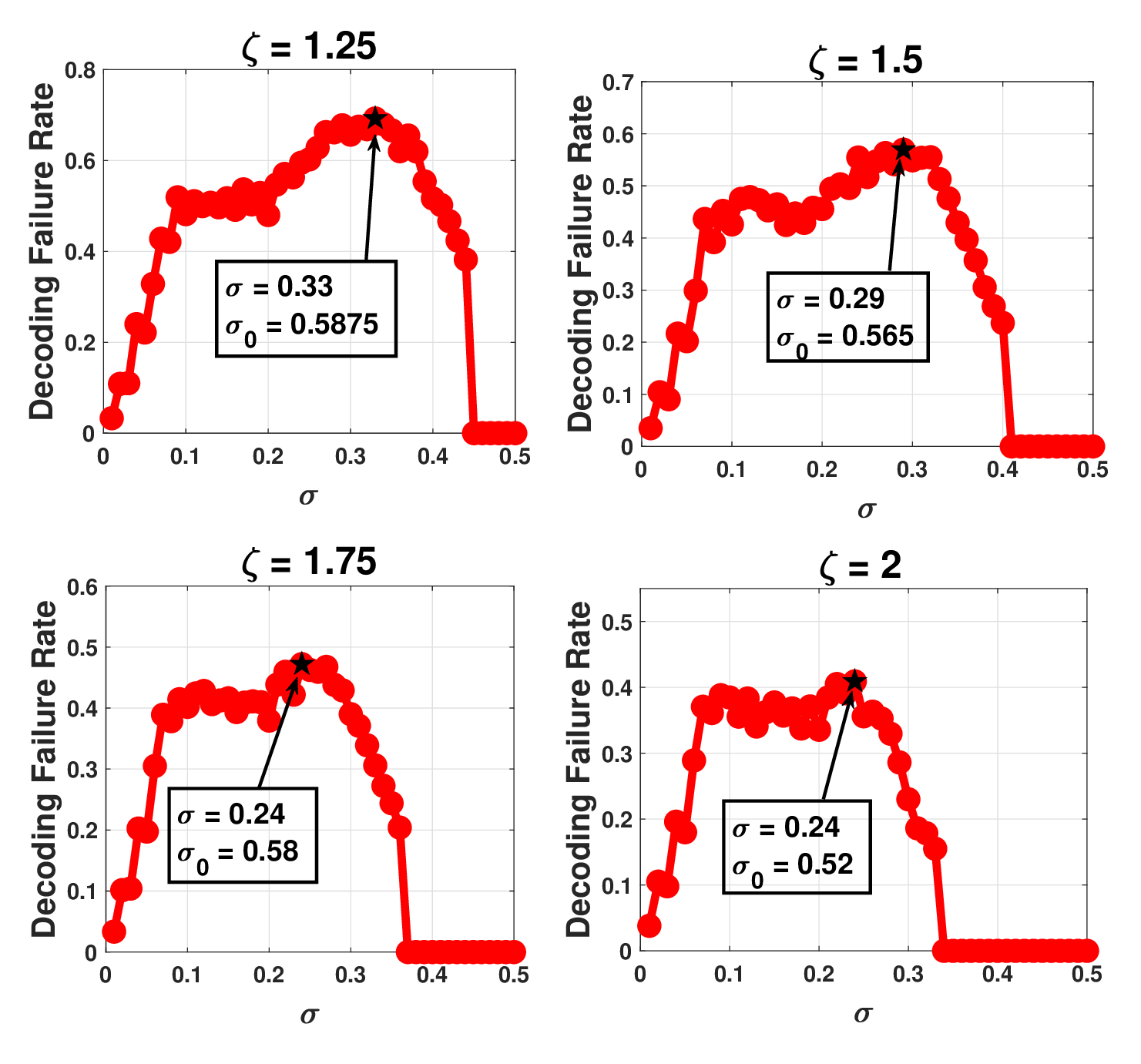}
  \vspace{-1cm}
  \caption{\textcolor{black}{Optimal ($\sigma_{0},\sigma$) pair to attack BP decoder for ($k=20$ and $S=60$) when $\nu=1\times S$.}}
  \label{fig:Optimization_attack_BP}
\end{figure}

In the previous section, we have characterized the decoding failure rates of various LT decoders for arbitrary values of $\sigma_{0}$ and $\sigma$ such that $\sigma_{0} \geq \sigma$. However, in practice, there are costs associated with reading and erasing droplet nodes. In this regard, we denote \textcolor{black}{$c_{r}>0$} and \textcolor{black}{$c_{e}>0$} as costs to read and erase one droplet node, respectively. Therefore, we introduce attack-cost as a new metric of interest, which is defined as follows. 

\begin{definition}[Attack-cost]\label{Attack Cost}
Attack-cost incurred by an adversary is defined as $c_{r}\times$\{no: of droplet nodes read\} + $c_{e}\times$\{no: of droplet nodes erased\}.
\end{definition}

In practice, $c_{r}$ is quantified as the time taken by an adversary to contact one droplet node and read its contents and $c_{e}$ is quantified as the time taken to erase one droplet node. Therefore, the attack-cost is quantified as the total time taken by an adversary to launch the attack. Typically, the attack-cost is upper-bounded in practice, to which we propose a method to choose the optimal pair ($\sigma_{0}$,$\sigma$) using which the cost constrained adversary can incur maximum failure rates on the LT decoders. 

\begin{figure}
  \includegraphics[scale = 0.32]{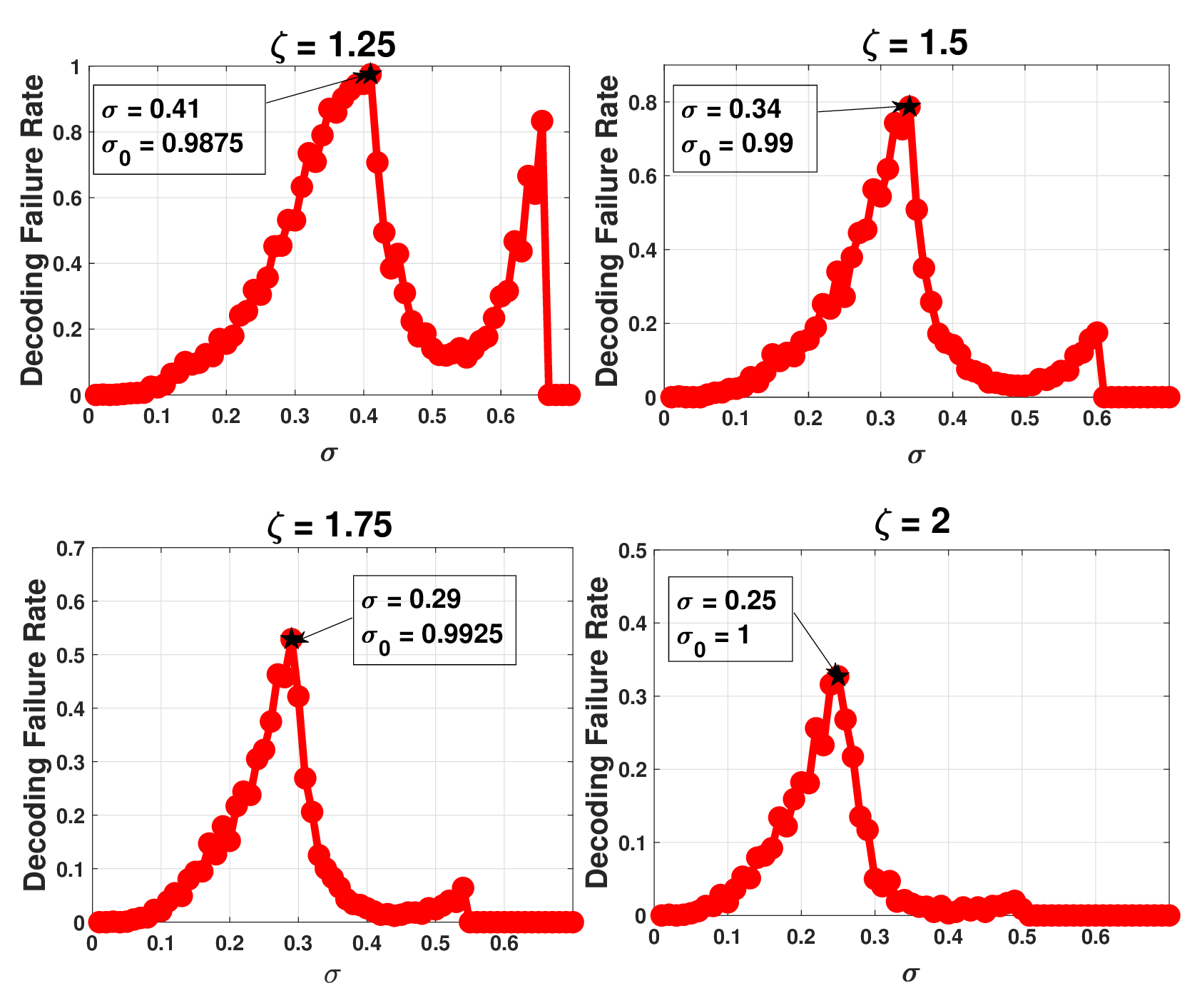}
  \vspace{-0.9cm}
  \caption{\textcolor{black}{Optimal ($\sigma_{0},\sigma$) pair to attack OFG decoder for ($k=20$ and $S=60$) when $\nu=1.5\times S$.}}
  \label{fig:Optimization_attack_OFG}
\end{figure}




By defining the read-write cost ratio as $\zeta \triangleq \frac{c_{e}}{c_{r}}$ we can rewrite the attack-cost as $c_{r}\times \sigma_{0} S + c_{e}\times \sigma S$, which could be simplified as $\sigma S c_{r} \times (\xi + \zeta)$.

To further simplify the analysis, we normalise $c_{r}$ to 1, and use the normalised attack-cost as $\sigma S \times (\xi + \zeta)$. 
Now, let the normalised attack-cost be upper-bounded by $\nu$ units for some $\nu > 0$. The corresponding inequality is
\vspace{-1mm}
\begin{equation}\label{eqn:Constraint_on_attack_cost}
    \sigma S \cdot(\xi + \zeta) \le \nu.
\end{equation}
In practice, given $S$, $\nu$ and $\zeta$, the adversary would need to choose the pair ($\sigma_{0}$,$\sigma$) such that the failure rates of the targeted decoder is maximized under the constraint in \eqref{eqn:Constraint_on_attack_cost} along with the constraints $0 \leq \sigma_{0} \leq 1$ and $0 \leq \sigma \leq \sigma_{0}$. Towards reducing the search space for the above problem without compromising on the failure rates, the following theorem presents tighter bounds on $\sigma$ and $\sigma_{0}$.

\begin{theorem}\thlabel{tight_bounds}
The feasible range for the search space on $\sigma$ and $\sigma_{0}$ is to use $\sigma \in (0, \frac{\nu}{S(1+\zeta)}]$, and then to fix $\sigma_{0} = \frac{\nu-\sigma S \zeta}{S}$ for a given value of $\sigma$.
\end{theorem}

\begin{proof}
Rearranging \eqref{eqn:Constraint_on_attack_cost}, we get $\sigma \cdot(\xi + \zeta)\le \frac{\nu}{S}$. The maximum value of $\sigma$ as a function of $\zeta$ under this constraint is $\sigma_{max}(\zeta) = \frac{\nu}{S(1+\zeta)}$. As a result, we get a constraint on $\sigma$ as $0 < \sigma \le \sigma_{max}(\zeta)$. Now, for each value of $\sigma$ in this range, the corresponding values of $\xi$ lies in the range $1\le \xi \le \frac{\nu-\sigma S \zeta}{\sigma S}$, where the upper limit on $\xi$ is obtained by rearranging \eqref{eqn:Constraint_on_attack_cost}. From our previous experiments on proving the effectiveness of the optimal attack strategies specific to various decoders, we have shown that for any decoder, the failure rate increases with $\xi$ when $\sigma$ is constant. Therefore, in this case for a given $\sigma$, the maximum value of $\xi$ that satisfies the constraint in \eqref{eqn:Constraint_on_attack_cost} results in the maximum decoding failure rate for that $\sigma$. 
\end{proof}

\begin{figure}[ht!]
  \includegraphics[scale = 0.325]{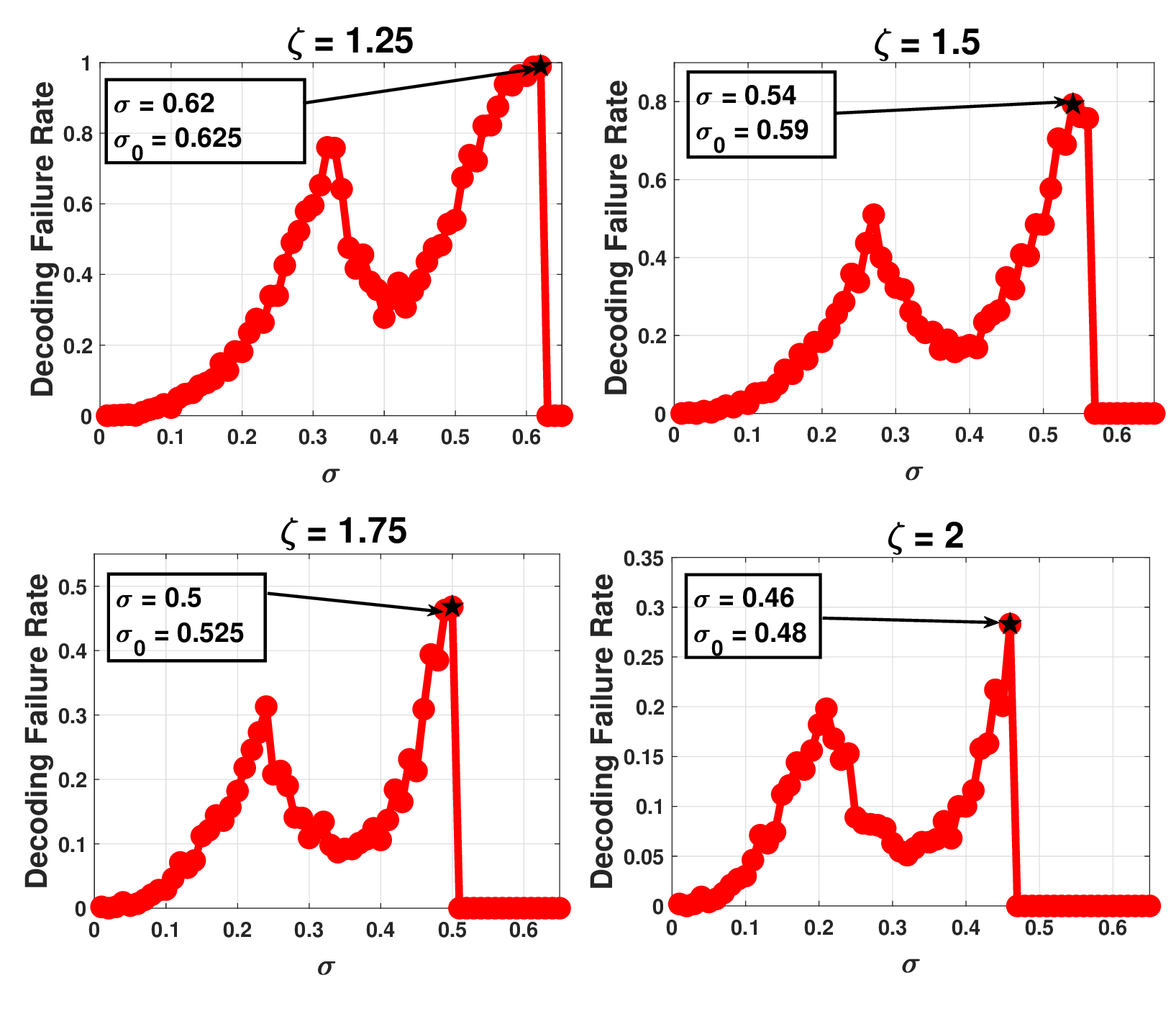}
  \vspace{-0.9cm}
  \caption{\textcolor{black}{Optimal ($\sigma_{0},\sigma$) pair to attack BRH decoder for ($k=20$ and $S=60$) with $K = 44$ when $\nu=1.4\times S$.}}
  \label{fig:Optimization_attack_BRH}
\end{figure}

Using \thref{tight_bounds}, we are able to reduce the search space to one-dimension, thereby reducing the attack complexity. 
Fig. \ref{fig:Optimization_attack_BP} to
Fig. \ref{fig:Optimization_attack_CRH} demonstrates the evaluation of decoding failure rates of BP, OFG, BRH and CRH decoders respectively, for values of $\sigma$ in the range $0 < \sigma \le \sigma_{max}(\zeta)$ with a step size of 0.01, and their corresponding maximum values of $\xi$. All four decoders use $k$ and $S$ values of 20 and 60, respectively. Fig. \ref{fig:Optimization_attack_BP} finds the best ($\sigma_{0}$,$\sigma$) using the score based attack strategy for BP decoder, wherein the normalised attack-cost of the adversary is constrained to $\nu=1\times S$. Fig. \ref{fig:Optimization_attack_OFG} is for the optimal OFG attack, where $\nu = 1.5\times S$. Fig. \ref{fig:Optimization_attack_BRH} corresponds to failure rates resulting from attacking the BRH decoder with $K = 44$ and $\nu = 1.4\times S$. Fig. \ref{fig:Optimization_attack_CRH} shows failure rates resulting from CRH decoder with $\eta_{c} = 6$ and $\nu = 1\times S$. 
\textcolor{black}{Also, we have used $\zeta\in \{1.25,1.5,1.75,2\}$ for all the experiments pertaining to cost constrained attacks. Note that we have used the values of $c_{e}$ and $c_{r}$ such that $c_{e} \geq c_{r}$, i.e., $\zeta \geq 1$. This condition stems from our assumption that the costs incurred for executing a DoS threat on a droplet node is typically much higher than reading the contents of a droplet node. Otherwise, using the values of $\zeta < 1$ would imply that the coded blockchain is weak enough to encourage adversarial entities to erase the contents of a droplet node with efforts much lower than that of downloading the contents of a droplet node. This is not a desired feature of a secure coded blockchain network.}

\begin{figure}
  \includegraphics[scale = 0.33]{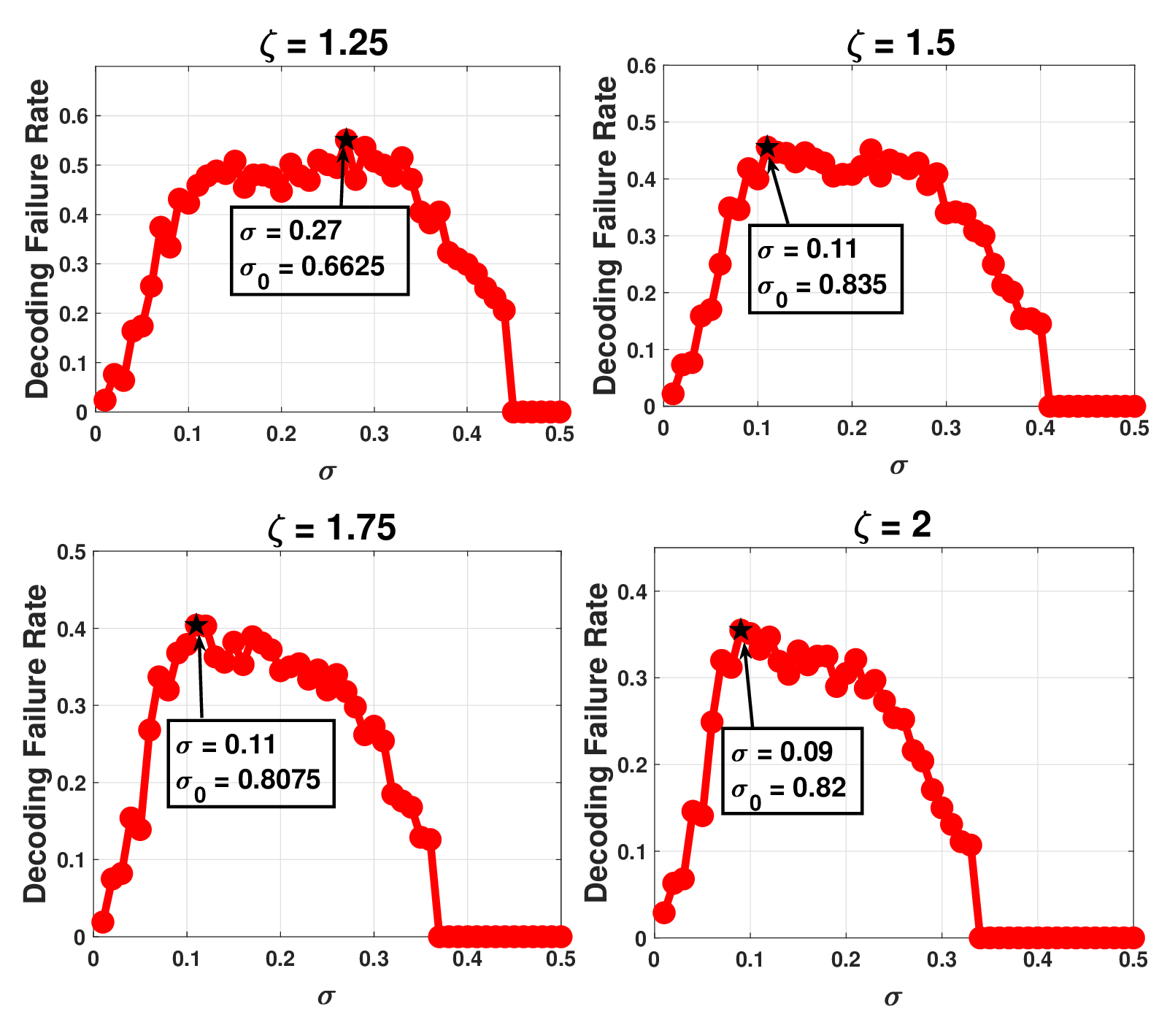}
  \vspace{-0.9cm}
  \caption{\textcolor{black}{Optimal ($\sigma_{0},\sigma$) pair to attack CRH decoder for ($k=20$ and $S=60$) with $\eta_{c} = 6$ when $\nu=1\times S$.}}
  \label{fig:Optimization_attack_CRH}
\end{figure}
\vspace{-1mm}
\subsection{Discussion}

In this section, we have highlighted the vulnerabilities of the LT-coded blockchain architecture against non-oblivious DoS attacks. Along these lines, we have proposed reasonable attack strategies specific to various LT decoder types which enables an adversary to prevent new nodes with specific download and computing capabilities from joining the blockchain network, limiting its scalability. Furthermore, we have also studied a cost-constrained adversary, in which the entire cost of accessing the storage nodes' contents and launching DoS attacks on a subset of them, is fixed. Under such constraints, our analysis assists the adversary in determining the optimal fraction of the storage nodes that must be read/erased in order to incur the highest decoding failure rate in the decoder he targets. 
Overall, we emphasize that when implementing LT coded blockchains, one must be aware of these threats, and hence take preventive measures against them.

\vspace{-2mm}
\section{Conclusion and Future Scope}
\vspace{-2mm}
This paper investigates two major challenges of LT coded blockchains: (i) Scaling LT coded blockchains in heterogeneous networks and (ii) Vulnerabilities of LT coded blockchain architecture to optimized non-oblivious DoS attacks. 

In the first part of the paper, we have presented BRH and CRH decoders to facilitate seamless scalability of LT coded blockchains in heterogeneous networks. It was demonstrated that the optimal mirroring costs of the BRH and CRH decoders are lower than those of the traditional BP and OFG decoders. While we highlight that the robust soliton distribution has been used at the encoding stage to support the BRH and CRH decoders, we believe that there is scope to design new degree distributions that are suited for the aforementioned classes of hybrid decoders. Hence, we leave the design and analysis of optimal degree distribution suitable for BRH and CRH decoders as a future work.

In the second part, we have investigated the vulnerabilities of LT-coded blockchains against a wide range of non-oblivious DoS attacks. We have identified that the degree of the droplet nodes, which are easily accessible to the adversary, open doors to the proposed class of attacks. \textcolor{black}{The attacks proposed in Section \ref{sec:threats} have carefully leveraged on the underlying algebraic structure of the BP and the OFG decoders of LT codes. While we were successful in convincing that these strategies make a significant impact on scalability, further studies on optimal attack strategies under the class of non-oblivious attacks are interesting directions for future research.} When adopting coded blockchains, it is critical to be aware of these threats and take precautions to mitigate them. Furthermore, as an interesting direction for future research, new coded architectures can be proposed for blockchains that enhances its resilience against optimized threats while also providing the storage savings and scalability features. \textcolor{black}{The novelty of our threat model lies in the choice of the droplet nodes for executing the DoS attack by the attacker. Therefore, we used simulation based experiments to validate the impact of the proposed contributions. However, if one is interested in measuring other peripheral metrics such as delays associated with implementing the attacks, and the impact of synchronization issues between the stages of attack and blockchain reconstruction, then these attacks could also be implemented on a hardware setup over a network of servers. These are interesting directions for future research. }

\vspace{-2mm}
\section{Acknowledgement}
\vspace{-2mm}
This work was supported by the research grant ``Secure Networks and Edge-Computing Hardware for Industry 4.0,” funded by the Ministry of Electronics and Information Technology, New Delhi, India.







\end{document}